\documentclass{fundam}

\publyear{22}
\papernumber{2102}
\volume{185}
\issue{1}

\usepackage{url} % takes care of hyperlinks, preferred over hyperref
\usepackage[ruled,lined]{algorithm2e}% provides Algorithm environment
\usepackage{graphicx}% allows for inclusion of graphic files (figures)
\usepackage{tikz}

\usepackage{amssymb}

\usetikzlibrary{automata, positioning, arrows}
\usepackage[table,xcdraw]{xcolor}
\usepackage[verbose]{hyperref}
\hypersetup{colorlinks=false,allbordercolors=blue,pdfborderstyle={/S/U/W 1}}
\usepackage{fix-cm}
\usepackage{caption}    % required by subcaption
\usepackage{subcaption} % provides the subfigure environment

\newcommand{\safeTitle}[2]{\texorpdfstring{#1}{#2}}

\begin{document}

\title{Partitioning $\mathbb{Z}_{sp}$ in finite fields and groups of trees and cycles}

%\address{Address for correspondence goes here}
\address{University of Ioannina,\\
Department of Informatics and Telecomunications,\\
Kostakioi Artas, 47100, Arta, Greece\\
}

\author{Nikolaos Verykios\\
Department of Informatics and Telecomunications, \\
University of Ioannina \\ Kostakioi Artas, 47100, Arta,\\ Greece\\
n.verykios{@}uoi.gr
\and Christos Gogos\\
Department of Informatics and Telecomunications, \\
University of Ioannina \\ Kostakioi Artas, 47100, Arta,\\ Greece\\
}

\maketitle

\runninghead{N. Verykios, C. Gogos}{Partitioning $\mathbb{Z}_{sp}$ in finite fields and groups of trees and cycles}

\begin{abstract}
This paper investigates the algebraic and graphical structure of the ring $\mathbb{Z}_{sp}$, with a focus on its decomposition into finite fields, kernels, and special subsets. We establish classical isomorphisms between $\mathbb{F}_s$ and $p\mathbb{F}_s$, as well as $p\mathbb{F}_s^{\star}$ and $p\mathbb{F}_s^{+1,\star}$. 

We introduce the notion of arcs and rooted trees to describe the pre-periodic structure of $\mathbb{Z}_{sp}$, and prove that trees rooted at elements not divisible by $s$ or $p$ can be generated from the tree of unity via multiplication by cyclic arcs. Furthermore, we define and analyze the set $\mathbb{D}_{sp}$, consisting of elements that are neither multiples of $s$ or $p$ nor “off-by-one” elements, and show that its graph decomposes into cycles and pre-periodic trees. 

Finally, we demonstrate that every cycle in $\mathbb{Z}_{sp}$ contains inner cycles that are derived predictably from the cycles of the finite fields $p\mathbb{F}_s$ and $s\mathbb{F}_p$, and we discuss the cryptographic relevance of $\mathbb{D}_{sp}$, highlighting its potential for analyzing cyclic attacks and factorization methods. \end{abstract}


\section{Introduction}\label{intro}
Finite fields with a prime number of elements have been extensively studied in the number theory literature~\cite{rogers1996graph}, \cite{kvrivzek2004sophie}, \cite{somer2006structure}. This paper extends relevant findings to the finite ring with $sp$ elements, where $s$ and $p$ are distinct primes. This ring has a crucial role in cryptography, serving as the algebraic framework for encryption and decryption computations. In this paper we consider primes $s$ and $p$ of the form $s=2^kq+1$ and $p=2^lr+1$, where $q$ and $r$ are odd numbers. The quadratic map $f(x)=x^2$ in the finite fields $p\mathbb{F}_s$ and $s\mathbb{F}_p$ and in the ring $\mathbb{Z}_{sp}$ produces directed graphs composed of cycles with trees attached to each cycle vertex as described by Rogers in~\cite{rogers1996graph}. For the specific case of finite fields having a prime number of elements, these trees are binary~\cite{somer2004connection}, \cite{rogers1996graph}. The set $p\mathbb{K}_s$ forms a binary tree rooted at unity, while the remaining elements constitute the set $p\mathbb{F}^{**}_s=p\mathbb{F}_s-p\mathbb{K}_s-\{0\}$. Thus, a finite field is partitioned into three distinct sets: the set $\{0\}$, the kernel $p\mathbb{K}_s$ of the mapping $f^k$ ($s=2^kq+1$) which is a multiplicative group, and the set of the remaining elements $p\mathbb{F}^{**}_s$.\\

The structure of the paper follows. In Section~\ref{sec:sec2} we prove that $/ps\mathbb{F}_s$ is isomorphic to $\mathbb{F}_s$ and we define the set $p\mathbb{F}^{+1}s=\{xp+1\mid x\in \mathbb{F}_s\}$ which is proved to also be isomorphic to $p\mathbb{F}_s$. In Section~\ref{sec:sec3} we prove that $s\mathbb{F}_p\times p\mathbb{F}_s$ is a ring isomorphic to $\mathbb{Z}_{sp}$ under function  $h$, by analyzing the cycles and the trees that appear in each subset of the Cartesian product $(\{0\}\cup s\mathbb{K}_p \cup s\mathbb{F}^{\ast\ast}_p)$~$ \times (\{0\}\cup p\mathbb{K}_s\cup p\mathbb{F}^{\ast\ast}_s)$. In Subsections~\ref{sec:sec3_1} and \ref{sec:sec3_2}, the tree of the unity of $h(\mathbb{K}_{sp})=s\mathbb{K}_p\times p\mathbb{K}_s$ and the cycles of $\mathbb{Z}_{sp}$ are explored, respectively. In Subsection \ref{sec:sec3_3} the concept of an arc in $\mathbb{Z}_{sp}$, which was introduced in our previous publication~\cite{verykios2024structures}, is used and further extended by defining the multiplication operation of an arc by a tree. These new concepts lay the ground for proving that all tree structures of $\mathbb{G}_{sp}$ are graphically identical to the tree of the quadratic function over the kernel of $\mathbb{Z}_{sp}$. Section~\ref{sec:sec3}, closes with Subsection~\ref{sec:sec3_4} which explores the structure of $s\mathbb{F}_p \times p\mathbb{K}_s$ and $s\mathbb{K}_p \times p\mathbb{F}_s$. Next, in Section~\ref{sec:sec4} the set $\mathbb{D}_{sp}=h^{-1}(s\mathbb{F}^{**}_p \times p\mathbb{F}^{**}_s)$ is defined and a proof is given about the fact that cycles included in this set are derived from the cycles of $s\mathbb{F}_p$ and $p\mathbb{F}_s$, that appear as inner cycles in their product. Section \ref{sec:sec5}  unifies the algebraic structures introduced earlier, defines the cryptographic set $\mathbb{D}{sp}$, and examines its cyclic and tree formations within $\mathbb{Z}{sp}$. The relation of these structures to RSA security and factorization methods is also discussed. Finally, the conclusion presents the key ideas of the paper and states the possible existence of novel factorization methods based on these cycles and trees.

\subsection{Related work}
Since the advent of RSA cryptography, research into finite fields has increased. Roger in his seminal work in~\cite{rogers1996graph} studied the finite fields of primes and their structure. A decade later, Somer and K{\v{r}}{\'\i}{\v{z}}ek in~\cite{somer2004connection} examined the fixed points of the quadratic function, alongside with the number of components and the length of cycles in finite fields.
Moreover, in finite fields with $p^n$ elements, where $p$ is a prime number and $n$ is a positive integer, Mullen and Vaughan \cite{mullen1988cycles} studied the cycle structures of polynomials of degree $n$ with coefficients from $\mathbb{F}_p$. Their work revealed several properties (i.e., dynamics), such as the number of cycles, the average preperiod length, and the cycle decomposition of the graph $G_f(\mathbb{F}_{p^n})$. They have also studied the fixed points of $\mathbb{F}_{p^n}$ which are relevant to cryptography. In recent years, Panario and Reis~\cite{panario2019functional} extended~\cite{mullen1988cycles} and fully described the functional graph of polynomials that are irreducible factors of $f(x)= x^n-1$. In their work they also described the graph $G_f(\mathbb{F}_{p^n})$, in detail.

The Chebyshev~\cite{mason2002chebyshev} and Dickson~\cite{lidl1993dickson} polynomials are the subject of much research in recent years. Cyclic structures of finite fields $\mathbb{F}_{p^2}$ have been studied by Lidl and Mullen in~\cite{as1991cycle} and many properties have been described in their work. In~\cite{qureshi2019graph}, Qureshi and Panario fully described the functional graph associated with iterations of Chebyshev polynomials over finite fields. They also used structural findings to obtain estimates for the average rho length, the average number of connected components, and the expected value for the period and preperiod of iterated Chebyshev polynomials. Finally, iterations of Rédei functions over non-binary finite fields have been studied in \cite{qureshi2017cycle} and \cite{panario2023construction}.

The main properties of cycles and trees in $\mathbb{D}_{sp}$ have been widely used to attack cryptographic algorithms and especially RSA. One of the earliest attack methods is the cyclic attack, which was firstly introduced by Simons and Norris in \cite{simmons1977preliminary} and later was generalized by Williams and Smitch in \cite{williams1979some} and Berkovits in \cite{berkovits1982factoring}. These algorithms are based on the early identification of the cycles of $\mathbb{Z}_{sp}$, and are considered to be less efficient~\cite{rivest1978remarks}. Another family of algorithms use various techniques related to the quadratic sieve \cite{pomerance1984quadratic}. These algorithms are essentially based on properties of quad trees, which is a topic that this paper and our previous one~\cite{verykios2024structures}, further explores. In quadratic sieve methods, all numbers $x$, $y$ that can produce a factorization of $N = sp$ are in the first level of the rooted trees on the cyclic elements of $\mathbb{D}_{sp}$. For these numbers it holds that  $y^2-x^2 = kN$, where $k\geq 1$. 

\section{\safeTitle{The fields $s\mathbb{F}_p$, $p\mathbb{F}_s$ and the groups $s\mathbb{F}_p^{\pm 1}$, $p\mathbb{F}_s^{\pm 1}$}{The fields sFp, pFs and the groups Fp+-1, Fs+-1}}
\label{sec:sec2}
Let $s=2^kq+1$ and $p=2^lr+1$ be prime and $q$ and $r$ be odd. It is known from Euclid and Bézout's identity that there are two integers $\alpha$ and $\beta$ such that $1 = \alpha s + \beta p$ and $\alpha \beta < 0$. Let us suppose that $\alpha$, the coefficient of $s$, is the negative integer. Then, the previous equality, can be written as $1 = -\alpha s + \beta p$, considering that both coefficients $\alpha$ and $\beta$ are positive numbers, and this is how they will be used for the rest of the paper. The values of $\alpha $ and $\beta $ can be determined using the extended Euclidean algorithm, and it holds that $-\alpha=s^{-1} \pmod p\in \mathbb{F}_p$ and $\beta =p^{-1}\pmod s\in \mathbb{F}_s$, where $\mathbb{F}_p$ and $\mathbb{F}_s$ are finite fields, each with a prime number of elements. In our previous work \cite{verykios2024structures} we have showed that $(\beta p)^2\equiv\beta p\pmod {sp}$ and $(-\alpha s)^2\equiv-\alpha s\pmod {sp}$ are the two fixed points of $\mathbb{Z}_{sp}$. Furthermore, let us consider the following subsets of $\mathbb{Z}_{sp}$:
\begin{itemize}
\item
$s\mathbb{F}_p=s.\mathbb{F}_p=\{xs\mid x\in [0..p-1]\}$ 
\item
$p\mathbb{F}_s=p.\mathbb{F}_s=\{yp\mid y\in [0..s-1]\}$ 
\item
$s\mathbb{F}_p^{+1}=\{xs+1\mid x\in [0..p-1]\}$
\item
$s\mathbb{F}_p^{-1}=\{xs-1\mid x\in [1..p]\}$
\item
$s\mathbb{F}_p^{\pm 1}=s\mathbb{F}_p^{+1}\cup s\mathbb{F}_p^{-1}=\{xs\pm 1\mid x\in [1..p]\}$
%\item $s\mathbb{F}^{\times}_p=$
\item
$p\mathbb{F}_s^{+1}=\{yp+1\mid y\in [0..s-1]\}$ 
\item
$p\mathbb{F}_s^{-1}=\{yp-1\mid y\in [1..s]\}$
\item
$p\mathbb{F}_s^{\pm 1}=p\mathbb{F}_s^{+1}\cup p\mathbb{F}_s^{-1}=\{yp\pm 1\mid y\in [1..s]\}$
\item
$\mathbb{G}_{sp}=\mathbb{Z}^{\times}_{sp}=\mathbb{Z}_{sp}-p\mathbb{F}_s-s\mathbb{F}_p$ 
%\item $\mathbb{D}_{sp}=\mathbb{Z}_{sp}-p\mathbb{F}_s-s\mathbb{F}_p-s\mathbb{F}_p^{\pm 1}-p\mathbb{F}_s^{\pm 1}$ 
\end{itemize}

The next two theorems prove that $p\mathbb{F}_s$ and $s\mathbb{F}_p$ are isomorphic fields of $\mathbb{F}_s$ and $\mathbb{F}_p$ respectively. Note that both theorems have been proved in our previous work \cite{verykios2024structures} but here we present an alternate proof that uses the method of calculating the inverse.

%\begin{theorem}
%\label{TheoremFour:pFs_field}
%Let $p$ and $s$ be prime positive integers. Define the set
%$p\mathbb{F}_s = \{0, p, 2p, \ldots, (s-1)p\} \subseteq \mathbb{Z}_{sp}$.
%Then subset $p\mathbb{F}_s$ of the ring $\mathbb{Z}_{sp}$ is a field, and it is isomorphic to the finite field $\mathbb{F}_s$.
%\end{theorem}

Define the map $g : \mathbb{F}_s \to p\mathbb{F}_s \subseteq \mathbb{Z}_{sp}, \quad g(x) = \beta p x \pmod {sp}$ where $ \beta \in \mathbb{Z}$ satisfies $\beta= p^{-1}\mod s$.
We will show that $g$ is a field isomorphism.
The mapping $g(x) = \beta p x \mod sp$, where $\beta p \equiv 1 \mod s$, induces a permutation of the set $p\mathbb{F}_s$, since the action of $\beta$ is invertible over $\mathbb{F}_s$. This observation relies on the well-known fact that multiplication by an invertible element modulo s acts as a permutation on $\mathbb{F}_s$.

\begin{theorem}
\label{TheoremFour:pFs_field}
Let $p$ and $s$ be distinct prime numbers. Consider the ring $\mathbb{Z}_{sp}$ and define the subset
\[
p\mathbb{F}_s = \{0, p, 2p, \ldots, (s-1)p\} \subseteq \mathbb{Z}_{sp}.
\]
Then $p\mathbb{F}_s$ admits a field structure isomorphic to the finite field $\mathbb{F}_s$, via the map
\[
g : \mathbb{F}_s \to p\mathbb{F}_s, \quad g(x) = \beta p x \mod sp,
\]
where $\beta \in \mathbb{Z}$ satisfies $\beta p \equiv 1 \mod s$.
\end{theorem}

\begin{proof}
Since $\gcd(p, s) = 1$, by Bézout's identity there exist integers $\alpha, \beta$ such that $-\alpha s + \beta p = 1$, where $\alpha$, $\beta$ are both positive.
Define the map $g : \mathbb{F}_s \to p\mathbb{F}_s \subseteq \mathbb{Z}_{sp}, \quad g(x) = \beta p x \pmod {sp}$.
We will show that $g$ is a field isomorphism.

Since $\beta p$ is coprime to $s$, the map $x \mapsto \beta p x \pmod {sp}$ is injective on $\mathbb{F}_s$, and the image consists of $s$ distinct elements of the form $\beta p x \pmod {sp}$, which are congruent to $px \pmod {sp}$. Thus, $g$ is bijective onto $p\mathbb{F}_s$.

Let $x, y \in \mathbb{F}_s$. Then $g(x + y) = \beta p (x + y) \equiv \beta p x + \beta p y = g(x) + g(y) \pmod {sp}$.
Similarly, $g(xy) = \beta p (xy)\pmod {sp}$.
But since $g(x) = \beta p x$, we also have $g(x) \cdot g(y) = (\beta p x)(\beta p y) = (\beta p)^2 xy$.
To reconcile this, observe that in $\mathbb{Z}_{sp}$, the scalar $(\beta p)^2 \equiv \beta p \pmod {sp}$ because $(\beta p)^2 = \beta p \cdot \beta p = \beta p (1 + \alpha s) \equiv \beta p \pmod {sp}$. Hence, $g(xy) \equiv g(x) \cdot g(y) \pmod {sp}$.

We compute $g(1) = \beta p \pmod {sp}$, which acts as the multiplicative identity in $p\mathbb{F}_s$, since for any $x \in \mathbb{F}_s$, $g(x) \cdot g(1) = (\beta p x)(\beta p) = (\beta p)^2 x \equiv \beta p x = g(x) \pmod {sp}$ (\cite{verykios2024structures}). 

Let $x \in \mathbb{F}_s^\times$. Then $x$ has an inverse $x^{-1} \in \mathbb{F}_s$, and $g(x) \cdot g(x^{-1}) = \beta p x \cdot \beta p x^{-1} = (\beta p)^2 \equiv \beta p \pmod {sp}$, so $g(x^{-1})$ is the multiplicative inverse of $g(x)$ in $p\mathbb{F}_s$.

Therefore, $g$ is a field isomorphism from $\mathbb{F}_s$ to $p\mathbb{F}_s$, and $p\mathbb{F}_s$ inherits the field structure from $\mathbb{F}_s$.

Note that the field operations on $p\mathbb{F}_s$ are defined via the isomorphism $g$, and do not coincide with the ring operations inherited from $\mathbb{Z}_{sp}$.
\end{proof}
The corresponding theorem for $s\mathbb{F}_p$ is stated and proved below:

\begin{theorem}
\label{TheoremFour:sFp_field}
Let $s$ and $p$ be distinct prime numbers. Consider the ring $\mathbb{Z}_{sp}$ and define the subset
\[
s\mathbb{F}_p = \{0, s, 2s, \ldots, (p-1)s\} \subseteq \mathbb{Z}_{sp}.
\]
Then $s\mathbb{F}_p$ admits a field structure isomorphic to the finite field $\mathbb{F}_p$, via the map
\[
g_1 : \mathbb{F}_p \to s\mathbb{F}_p, \quad g_1(x) = -\alpha s x \mod sp,
\]
where $\alpha \in \mathbb{Z}$ satisfies $-\alpha s \equiv 1 \mod p$.
\end{theorem}

\begin{proof}
Since $\gcd(p, s) = 1$, Bézout's identity guarantees the existence of integers $\alpha, \beta \in \mathbb{Z}$ such that $-\alpha s + \beta p = 1$. Define the map $g_1 : \mathbb{F}_p \to s\mathbb{F}_p \subseteq \mathbb{Z}_{sp}$ by $g_1(x) = -\alpha s x \mod sp$. Because $-\alpha s \equiv 1 \mod p$, the map $g_1$ is injective and surjective, and its image consists of $p$ distinct elements of the form $-\alpha s x \mod sp$, which are congruent to $sx \mod sp$. For any $x, y \in \mathbb{F}_p$, we have
\[
g_1(x + y) = -\alpha s (x + y) \equiv -\alpha s x + (-\alpha s y) = g_1(x) + g_1(y) \mod sp,
\]
and similarly $g_1(xy) = -\alpha s (xy) = -\alpha sxy\mod sp$.

Moreover,
\[
g_1(x) \cdot g_1(y) = (-\alpha s x)(-\alpha s y) = (-\alpha s)^2 xy,
\]
and since $(-\alpha s)^2 \equiv -\alpha s \mod sp$, it follows that $g_1(xy) \equiv g_1(x) \cdot g_1(y) \mod sp$. We compute $g_1(1) = -\alpha s \mod sp$, which acts as the multiplicative identity in $s\mathbb{F}_p$, since
\[
g_1(x) \cdot g_1(1) = (-\alpha s x)(-\alpha s) = (-\alpha s)^2 x \equiv -\alpha s x = g_1(x) \mod sp.
\]
For any $x \in \mathbb{F}_p^\times$, let $x^{-1}$ be its inverse in $\mathbb{F}_p$; then
\[
g_1(x) \cdot g_1(x^{-1}) = (-\alpha s x)(-\alpha s x^{-1}) = (-\alpha s)^2 \equiv -\alpha s = g_1(1) \mod sp,
\]
so $g_1(x^{-1})$ is the multiplicative inverse of $g_1(x)$. Therefore, $g_1$ is a field isomorphism from $\mathbb{F}_p$ to $s\mathbb{F}_p$, and $s\mathbb{F}_p$ inherits the field structure via $g_1$.
\end{proof}

It is known that the multiples of $s$ and $p$ have no inverse in $\mathbb{Z}_{sp}$. A strong result of these theorems is that if we reduce the domain of definition of the functions $g_1$ from the ring $\mathbb{Z}_{sp}$ to the finite domain $p\mathbb{F}_s$, then the elements $px$ have inverse (Theorem \ref{TheoremFour:pFs_field}). Furthermore, if we reduce the domain of the function $g_2$ to the finite field of multiples of $s$ (i.e., to the new domain $s\mathbb{F}_p$) then all the multiples of $s$ also have inverse.

It has been proved that $g_1$ is a multiplicative isomorphism mapping $\mathbb{F}_s$ to $p\mathbb{F}_s$. The group $p\mathbb{F}_s=p.\mathbb{F}_s$ is the set of $p$ multiples of $\mathbb{F}_s$, and we define the set $\beta p\mathbb{F}_s=\{\beta px\pmod {sp}\mid x\in \mathbb{F}_s\}$, which is actually a permutation of $p\mathbb{F}_s$. An analogous definition of a permutation of $s\mathbb{F}_p$ is the group $-\alpha s\mathbb{F}_p=\{-\alpha sx\pmod {sp}\mid x\in \mathbb{F}_p\}$.

The sets $\mathbb{F}_s$ and $\mathbb{F}_p$ are isomorphic to $p\mathbb{F}_s$ and $s\mathbb{F}_p$ respectively, and next we will prove that they are also isomorphic to $p\mathbb{F}_s^{+1}$ and $s\mathbb{F}_p^{+1}$. We will also prove that $p\mathbb{F}_s^{\pm 1}$ and $s\mathbb{F}_p^{\pm 1}$ are multiplicative groups.

We define the following sets: $\mathbb{F}_s^{\ast}=\mathbb{F}_s - \{0\}$ and $\mathbb{F}_p^{\ast}=\mathbb{F}_p - \{0\}$, as well as $p\mathbb{F}_s^{+1,\ast}=p\mathbb{F}_s^{+1} - \{-\beta p+1\}$ and $s\mathbb{F}_p^{+1,\ast}=s\mathbb{F}_p^{+1} - \{\alpha s+1\}$. Note that $-\alpha s+\beta p=1$. The elements $-\alpha s=-\beta p+1$ and $\beta p=\alpha s+1$ serve as the zero multiplicative elements of $p\mathbb{F}_s^{+1}$ and $s\mathbb{F}_p^{+1}$, respectively, as will be proved in the next two theorems. 

\begin{lemma}
\label{Theorem_pFs+1_closure_zeros}
Let $p\mathbb{F}_s^{+1,\ast} =\{\kappa p+1\mid \kappa \in [0..s-1] \setminus \{s-\beta\}\}= p\mathbb{F}_s^{+1} \setminus \{-\beta p+1\}$. Then:
\begin{enumerate}
\item The element $-\beta p+1$ is the unique multiplicative zero element in $p\mathbb{F}_s^{+1}$.
\item The set $p\mathbb{F}_s^{+1,\ast}$ is closed under multiplication.
\end{enumerate}
\end{lemma}

\begin{proof}

\begin{enumerate}
\item
% First, let us prove that the unique zero element of $s\mathbb{F}_p^{+1}$ is the element $\beta p = \alpha s+1$. It holds that $(\alpha s+1)(xs+1)=\beta p(xs+1)=\beta xps+\beta p=\beta p=\alpha s+1$ for any $xs+1$ element of $s\mathbb{F}_p^{+1}$. Let us suppose now that there is one more element $ts+1$ which behaves like a zero, i.e., for $xs+1\neq 1$ ($x\neq 0$), we have $(ts+1)(xs+1)\equiv (ts+1)\pmod{sp}$, $(txs+x)s\equiv 0\pmod {sp}$, $txs+x\equiv 0\pmod {p}$, $(ts+1)x\equiv 0 \pmod {p}$, which means that $ts+1=kp$ or $-ts+kp=1$. This is only true when $t=\alpha$ and $k=\beta$, because 1 can be written as a linear combination of $s$ and $p$ in the unique way $1=-\alpha s+\beta p$.
We use the Bézout identity $-\alpha s + \beta p = 1$ from paragraph \ref{sec:sec2}, hence $-\alpha s = -\beta p + 1$. Put $z=-\beta p+1$.
For any $xp+1\in p\mathbb{F}_s^{+1}$,
$z(xp+1) =(-\beta p+1)(xp+1) = (-\alpha s)(xp+1) = -\alpha s \cdot xp - \alpha s \equiv -\alpha s \pmod{sp}$,
since $-\alpha s \cdot xp$ is divisible by $sp$. As $-\alpha s = z$, we obtain $z(xp+1)\equiv z\pmod{sp}$ for every $x$, so $z$ is an absorbing (multiplicative zero) element.

Suppose $tp+1$ is another absorbing element, i.e. for all $x$,
$(tp+1)(xp+1)\equiv tp+1 \pmod{sp}$.
Subtracting $tp+1$ gives $p\big(txp+x\big)\equiv 0\pmod{sp}$, so $txp+x\equiv 0\pmod{s}$ for all $x$. Thus $x(tp+1)\equiv 0\pmod{s}$ for every $x$, which forces $tp+1\equiv 0\pmod{s}$. Hence $tp+1 = ks$ for some integer $k$, and reducing modulo $s$ yields
$tp \equiv -1 \pmod{s}$, or $t \equiv -p^{-1} \pmod{s}$.
But $p^{-1}\equiv \beta \pmod{s}$, so $t\equiv -\beta \equiv s-\beta\pmod{s}$. Therefore $tp+1 = (s-\beta)p+1 \equiv -\beta p+1 \pmod{sp}$, hence $tp+1 = -\beta p+1$ as elements of $p\mathbb{F}_s^{+1}$. This proves uniqueness.

\item
Let $xp+1$ and  $yp+1$ be elements of $p\mathbb{F}_s^{+1,\ast}$ and suppose that $(xp+1)(yp+1)\equiv -\beta p+1 \pmod{sp}$. Then $p(xyp+x+y+\beta)\equiv 0 \pmod{sp}$, hence $xyp+x+y\equiv -\beta \pmod{s}$. Since $p\equiv \beta^{-1}\pmod{s}$, this gives $xy\beta^{-1}+x+y\equiv -\beta\pmod{s}$, and multiplying by $\beta$ yields $(x+\beta)(y+\beta)\equiv 0\pmod{s}$. As $s$ is prime, this forces $x\equiv -\beta$ or $y\equiv -\beta$, contradicting the hypothesis. Thus the product never equals $-\beta p+1$.
\end{enumerate}
\end{proof}

\begin{theorem}
\label{Theorem_pFs+1_AbelianGroup}
The set $p\mathbb{F}_s^{+1,\ast}=\{\kappa p+1\mid \kappa \in [0..s-1] \setminus \{s-\beta\}\}$ is a multiplicative group.
\end{theorem}

\begin{proof}
Let $x,y\in\mathbb{F}_s$. Since $(xp+1)(yp+1)=(xyp+x+y)p+1\equiv tp+1\pmod{sp}$, where $t=(xyp+x+y)\bmod s$, the set $p\mathbb{F}_s^{+1}$ is multiplicatively closed, and associativity and commutativity hold.

We now establish that $z = -\beta p+1 = -\alpha s$ is the unique multiplicative zero element. For any $xp+1 \in p\mathbb{F}_s^{+1}$, we have $z(xp+1) = -\beta xp^2 - \beta p + xp + 1$. Working modulo $s$, since $p \equiv \beta^{-1} \pmod{s}$, we get $-\beta xp + x \equiv -x + x \equiv 0 \pmod{s}$. Thus $-\beta xp^2 + xp = ks$ for some integer $k$, giving $z(xp+1) \equiv ks + z \equiv z \pmod{sp}$. For uniqueness, if $(tp+1)(xp+1) \equiv tp+1 \pmod{sp}$ for all $x$, then taking $x=1$ yields $p(tp+1) \equiv 0 \pmod{sp}$, so $tp+1 \equiv 0 \pmod{s}$. This gives $-tp+ks=1$, which by uniqueness of Bézout coefficients implies $t \equiv -\beta \pmod{s}$, hence $tp+1 = -\beta p+1$.

The identity element is $1 = 0 \cdot p + 1$. For any $xp+1 \in p\mathbb{F}_s^{+1,\ast}$ with $x \neq s-\beta$, the inverse $yp+1$ satisfies $(xp+1)(yp+1) \equiv 1 \pmod{sp}$, which expands to $(xyp+x+y)p+1 \equiv 1 \pmod{sp}$. This requires $y(xp+1)+x \equiv 0 \pmod{s}$, yielding $y \equiv -x(xp+1)^{-1} \pmod{s}$. Since $x \neq s-\beta$ ensures $xp+1 \not\equiv 0 \pmod{s}$, the inverse exists uniquely.
\end{proof}

We define the set: $p\mathbb{F}_s^{\pm 1,\ast}=p\mathbb{F}_s^{\pm 1} - \{-\beta p+1, \beta p-1\}$. %Note that $(-\beta p+1)^2=(-\beta p)^2-2\beta p+1=\beta p-2\beta p+1=-\beta p+1$, and $(\beta p-1)^2=(\beta p)^2-2\beta p+1=\beta p-2\beta p+1=-\beta p+1$.

We now extend the multiplicative structure established in
Theorem~\ref{Theorem_pFs+1_isomorphic} to the combined families $p\mathbb{F}_s^{+1}$ and $p\mathbb{F}_s^{-1}$.
In particular, we show that—after removing the two exceptional (non-invertible) elements—the set of all terms of the form $xp\!\pm\!1$ is closed under multiplication and forms a multiplicative group.

\begin{theorem}
\label{Theorem_pFs+_1_AbelianMultiplicativeGroup}
The set $p\mathbb{F}_s^{\pm1,\ast}
= p\mathbb{F}_s^{\pm1} \setminus \{-\beta p + 1,\, \beta p - 1\}$ is a multiplicative group.
\end{theorem}

\begin{proof}
Let $x, y \in \mathbb{F}_s \setminus \{\beta, -\beta\}$. Then 
\[(xp \pm 1)(yp \pm 1) = (xyp \pm x \pm y)p \pm 1 \equiv tp \pm 1 \pmod{sp},\]
where $t \equiv xyp \pm x \pm y \pmod{s}$. We show that $tp \pm 1 \not\equiv \pm(\beta p - 1) \pmod{sp}$, ensuring closure.

\smallskip
\noindent\textbf{(i)} We have already shown that $(xp+1)(yp+1) \equiv -\beta p + 1 \pmod{sp}$ when either $x$ or $y$ equals $-\beta$ (Lemma \ref{Theorem_pFs+1_closure_zeros}).
The congruence $(xp+1)(yp+1) \equiv \beta p - 1$ is impossible.

\smallskip
\noindent\textbf{(ii)} Suppose $(xp-1)(yp-1) \equiv -\beta p + 1 \pmod{sp}$ for
$x, y \notin \{\beta, s-\beta\}$. Expanding and reducing modulo $s$ using
$p \equiv \beta^{-1}$ gives
\[
xy\beta^{-2} - (x+y)\beta^{-1} \equiv -1 \pmod{s},
\]
or equivalently $(x-\beta)(y-\beta) \equiv 0 \pmod{s}$, implying $x \equiv \beta$ or $y \equiv \beta$, a contradiction.

\smallskip
\noindent\textbf{(iii)} Suppose $(xp-1)(yp+1) \equiv \beta p - 1 \pmod{sp}$ for some $x, y \in \mathbb{F}_s \setminus \{\beta, s-\beta\}$.
Then
\[
xyp^2 + (x-y)p \equiv \beta p \pmod{sp},
\]
which modulo $s$ yields $(x+\beta)(y-\beta) \equiv 0 \pmod{s}$.
Hence $x \equiv -\beta$ or $y \equiv \beta$, again impossible.
The symmetric case $(xp+1)(yp-1)$ is identical.

\smallskip
Thus closure holds. Commutativity and associativity follow from $\mathbb{F}_s$.
%The excluded elements $-\beta p + 1$ and $\beta p - 1$ act as zeros of $p\mathbb{F}_s^{+1}$ and $p\mathbb{F}_s^{-1}$, respectively.
The identity is $1 = 0\cdot p + 1$.

For inverses, the case $xp+1$ is handled as in
Theorem~\ref{Theorem_pFs+1_isomorphic}.
For $xp-1$, let its inverse be $yp-1$:
\[
(xp-1)(yp-1) \equiv 1 \pmod{sp}
\;\Rightarrow\;
xyp^2 - xp - yp \equiv 0 \pmod{sp}.
\]
Reducing modulo $s$ gives $xy - x - y \equiv 0$, hence
$y(x-1) \equiv x \pmod{s}$ and
$y \equiv x(x-1)^{-1} \pmod{s}$.
Since $x \notin \{\beta, s-\beta\}$, the inverse exists uniquely.

Therefore, every element has a unique multiplicative inverse, and $p\mathbb{F}_s^{\pm1,\ast}$ forms a multiplicative group.
\end{proof}

\begin{theorem}
\label{Theorem_pFs+1_isomorphic}
The groups $p\mathbb{F}_s^{+1,\ast}$ and $p\mathbb{F}_s^{\ast}$ are isomorphic.
\end{theorem}

\begin{proof}
Define $g : p\mathbb{F}_s^{+1,\ast} \to p\mathbb{F}_s^{\ast}$ by $g(xp+1) = \beta p(xp+1) \bmod sp$ for $x \in \mathbb{F}_s \setminus \{s-\beta\}$. 

We verify that $g(xp+1) \in p\mathbb{F}_s^{\ast}$. Clearly $g(xp+1) = \beta(xp^2 + p)$. Working modulo $s$:
$g(xp+1) \equiv \beta xp + \beta \pmod{s}$.
Since $p \equiv \beta^{-1} \pmod{s}$, we have
$g(xp+1) \equiv \beta x \beta^{-1} + \beta \equiv x + \beta \pmod{s}$.
Since $x \neq s-\beta$, we have $x + \beta \not\equiv 0 \pmod{s}$, so $g(xp+1) \not\equiv 0 \pmod{sp}$. Thus $g(xp+1) \in p\mathbb{F}_s^{\ast}$.

If $g(x_1p+1) = g(x_2p+1)$, then $\beta p(x_1p+1) \equiv \beta p(x_2p+1) \pmod{sp}$. Working modulo $s$: $\beta (x_1p+1) \equiv \beta (x_2p+1) \pmod{s}$, thus $\beta x_1p \equiv \beta x_2p \pmod{s}$. Since $\gcd(\beta,s)=1$ and $\gcd(p,s)=1$, this implies $x_1 \equiv x_2 \pmod{s}$, and thus $x_1p+1 = x_2p+1$. Since both groups have $s-1$ elements, injectivity implies $g$ is bijective. 

For $x_1p+1, x_2p+1 \in p\mathbb{F}_s^{+1,\ast}$, we have $g((x_1p+1)(x_2p+1)) = \beta p(x_1p+1)(x_2p+1)$ and $g(x_1p+1)g(x_2p+1) = (\beta p)^2(x_1p+1)(x_2p+1)$.
Since $(\beta p)^2 \equiv \beta p \pmod{sp}$ \cite{verykios2024structures}, we have $g((x_1p+1)(x_2p+1)) \equiv g(x_1p+1)g(x_2p+1) \pmod{sp}$. 

Hence $g$ is a group isomorphism.
\end{proof}

Similarly, we can prove that $s\mathbb{F}_p^{+1,\star}$, $s\mathbb{F}_p^{\pm 1,\star}$ are multiplicative groups and that $s\mathbb{F}_p^{+1,\star}$ is isomorphic to $s\mathbb{F}_p^{\star}$.

\section{A ring of a Cartesian product}\label{sec:sec3}

This section begins with the definition of a Cartesian product resulting in a set and continues with definition about addition and multiplication over this set and a proof that it is a ring. 

\begin{definition}
\label{Def_CartetianDefinition}
Let the Cartesian product $s\mathbb{F}_p\times p\mathbb{F}_s$, consisting of pairs $(xs,yp)$ and let the addition over elements of it be defined as $+:s\mathbb{F}_p\times p\mathbb{F}_s\mapsto s\mathbb{F}_p\times p\mathbb{F}_s$ such that $(x_1s,y_1p)+(x_2s+y_2p)=((x_1+x_2)\pmod p s,(y_1+y_2)\pmod s p)$. Also, let the multiplication $.:s\mathbb{F}_p\times p\mathbb{F}_s\mapsto s\mathbb{F}_p\times p\mathbb{F}_s$ be defined as $(x_1s,y_1p).(x_2s,y_2p)=(x_1x_2s^2\pmod{sp}, y_1y_2p^2\pmod{sp})=((x_1x_2s)\pmod p s,(y_1y_2p)\pmod s p)$. \\
\end{definition}

\begin{theorem}
\label{Theorem_cartesianRing}
The structure $(s\mathbb{F}_p\times p\mathbb{F}_s,\hspace{0.5em} +,\hspace{0.5em} .)$ is a ring with divisors of zero.
\end{theorem}

\begin{proof}
Most of the properties of the ring are easy to prove, so we will only show the proof of the most important ones, that are of interest for subsequent parts of this paper. We will start by proving that $(s\mathbb{F}_p\times p\mathbb{F}_s,+)$ is an Abelian group. The zero element is $(0,0)$ and the opposite of $(xs,yp)$ is the element $(-xs,-yp)=((p-x)s,(s-y)p)$.
The multiplication, as was defined above, guarantees that the Cartesian product is closed under multiplication. The unity of $s\mathbb{F}_p\times p\mathbb{F}_s$ is the element $(-\alpha s, \beta p)$, where $-\alpha s$ and $\beta p$ are the unities of $s\mathbb{F}_p$ and $p\mathbb{F}_s$ and they are the fixed points of $\mathbb{Z}_{sp}$. In rings, every non-zero element has an inverse, and in our case we can compute the inverse $(x_2s,y_2p)$ of $(x_1s,y_1p)$ by $(x_1s,y_1p)(x_2s,y_2p)=(-\alpha s,\beta p)$ which gives $x_2=-a(x_1s)^{-1}$ and $y_2=-a(y_1s)^{-1}$. 
\end{proof}
 
In the following theorem we show that the Cartesian product $s\mathbb{F}_p\times p\mathbb{F}_s$ is isomorphic to $\mathbb{Z}_{sp}$, and it will of much importance subsequently in this paper.

\begin{theorem}
\label{Theorem_mainIsomorphism_Zsp}
The rings $\mathbb{Z}_{sp}$ and $s\mathbb{F}_p\times p\mathbb{F}_s$ are isomorphic.
\end{theorem}

\begin{proof}
It is known that for two coprimes $s$ and $p$ there are positive $\alpha$ and $\beta$ such that $1=-\alpha s+\beta p$. Without loss of generality, we assume $s$ to be the prime with the minus sign and $p$ the prime with the positive sign in front of it. Multiplying the equality by $w\in \mathbb{Z}_{sp}$ we get $w=(-\alpha ws + \beta wp)\pmod{sp}=xs + yp$, where $x =-\alpha w \pmod p$ and $y =\beta w \pmod s$. A new function $h:\mathbb{Z}_{sp}\mapsto s\mathbb{F}_p\times p\mathbb{F}_s$ is now defined such that for every $w\in \mathbb{Z}_{sp}$ the image $h(w)$  is equal to $ (xs, yp) \in s\mathbb{F}_p\times p\mathbb{F}_s$.\\
Indeed $h$ is a function since for any two elements $w_1=-x_1s+y_1p$, $w_2=x_2s+y_2p$ of $s\mathbb{F}_p\times p\mathbb{F}_s$, assuming $w_1=w_2$ it holds that $(x_1 - x_2)s=(y_1-y_2)p$. Since $(s,p)=1$ it holds that $p$ divides $x_1-x_2$ and $s$ divides $y_1-y_2$. It is known that $x_1, x_2 \in \mathbb{F}_p$ so, $0<x_1<p$ and $-p<-x_2<0$ and therefore $-p<x_1-x_2<p$. Since, $p$ divides $x_1-x_2$ we conclude that $x_1-x_2=0$ or $x_1=x_2$. Likewise, we can prove that $y_1=y_2$. Therefore, it was proved that $h$ is a function, i.e., for each $w\in\mathbb{Z}_{sp}$ there exists a unique pair $(-xs, yp)\in s\mathbb{F}_p\times p\mathbb{F}_s$ such that $h(w)=(-xs, yp)$ or $w=-xs+yp$.\\
It is clear that for any $(-xs, yp)\in s\mathbb{F}_p\times p\mathbb{F}_s$ the sum $w=-xs+yp\in \mathbb{Z}_{sp}$, which means that function $h$ is onto. Then, we will prove that $h$ is one-to-one. It is easy to see that if $(-x_1s,y_1p)=(-x_2s,y_2p)$ then $-x_1s+y_1p=-x_2s+y_2p$ holds.\\
%Furthermore, it is obvious that the population of both $\mathbb{Z}_{sp}$ and $s\mathbb{F}_p\times p\mathbb{F}_s$ is equal to $sp$.\\
Based on $h$, we have indeed proved that for any positive integer $w$ less than $sp$ there are unique $x\in\mathbb{F}_p$ and $y\in\mathbb{F}_s$ such that $w=-xs+yp$, and so far there is no other way, than factoring the product $sp$, to find this decomposition.\\
We also need to make it clear that $h$ respects both operations ($+$, $.$), from one ring to the other. For addition, this holds because $h((x_1s+y_1p)+(x_2s+y_2p)) = h((x_1+x_2)s+(y_1+y_2)p) = ((x_1+x_2)s,(y_1+y_2)p)=(x_1s,y_1p)+(x_2s,y_2p)=h(x_1s+y_1p)+h(x_2s+y_2p)$. Likewise, for multiplication it holds that $h((x_1s+y_1p)(x_2s+y_2p))= h(x_1x_2s.s+y_1y_2p.p)= (x_1x_2s.s,y_1y_2p.p)= (x_1s,y_1p)(x_2s,y_2p)= h(x_1s+y_1p).h(x_2s+y_2p)$.\\
Finally, we need to prove that $h$ also preserves the unities of the rings. It has already been stated that $1=-\alpha s+\beta p$ and thus $h(1)=(-\alpha s,\beta p)$. Moreover, it is known that $-\alpha s$ and $\beta p$ are the fixed points of $s\mathbb{F}_p$ and $p\mathbb{F}_s$ and at the same time the unities of the corresponding fields (Theorems \ref{TheoremFour:pFs_field}, \ref{TheoremFour:sFp_field}).
\end{proof}

A consequence of the above theorems is the next corollary.
\begin{corollary}
\label{corollary_AllRingsIsomorphism}
Rings $\mathbb{Z}_{sp}$, $s\mathbb{F}_p\times p\mathbb{F}_s$ and $\mathbb{F}_p\times \mathbb{F}_s$ are isomorphic.
\end{corollary}
The isomorphism between $\mathbb{Z}_{sp}$ and $s\mathbb{F}_p\times p\mathbb{F}_s$ makes clear that every number of $\mathbb{Z}_{sp}$ can be written as a unique sum between a multiple of $s$ and a multiple of $p$ and vice versa \cite{verykios2024structures}. 
In the next sections we will study the relation of the cycles and the trees appearing in $\mathbb{Z}_{sp}$ to the cycles and the trees of the finite fields $s\mathbb{F}_p$ and $p\mathbb{F}_s$.

The quadratic function $f$ of the ring $\mathbb{Z}_{sp}$ (and of its subgroups) combined with the function $h$ of Theorem~\ref{Theorem_mainIsomorphism_Zsp}, is used to represent function $f$ in the isomorphic groups as follows:
$h(f(w))= h(f(xs+yp))=h((xs+yp)^2)=h(x^2s^2+y^2p^2)= (x^2s^2, y^2p^2)= (f(xs),f(yp))=h(f(xs),f(yp))$.

\subsection{\safeTitle{The tree structures of $\mathbb{Z}_{sp}$}{The tree structures of Zsp}}
\label{sec:sec3_1}

In this section we study homomorphisms of the quadratic function itself over finite fields and groups. Let $s=2^kq+1$ be a prime number, $k\geq 1$ and $q$ be odd while $f$ is  the quadratic function. The function $f$ is a homomorphism from $\mathbb{F}_s$ to $\mathbb{F}_s$, from $p\mathbb{F}_s$ to $p\mathbb{F}_s$ and from $\mathbb{Z}_{sp}$ to $\mathbb{Z}_{sp}$.

The set $Ker(f^n\mid_{A})=\{x\mid x^{2^n}=1\}$ is the kernel of the corresponding homomorphism over $A$, where $A$ is any of $\mathbb{F}_s$, $\mathbb{F}_p$, $p\mathbb{F}_s$, $s\mathbb{F}_p$ and $\mathbb{Z}_{sp}$. We use the notation  $\mathbb{K}_s$, $\mathbb{K}_p$, $p\mathbb{K}_s=Ker(f^k\mid_{p\mathbb{F}_s})$, $s\mathbb{K}_p=Ker(f^l\mid_{s\mathbb{F}_p})$ and $\mathbb{K}_{sp}=Ker(f^n\mid_{\mathbb{Z}_{sp}})$ in accordance with the notation used in~\cite{koblitz1994course}. Kernels of quadratic functions over finite fields such as $\mathbb{F}_s$ and $p\mathbb{F}_s$ (i.e., $s=2^kq+1$) have been shown to be binary trees (\cite{rogers1996graph}, \cite{mullen1988cycles}), where the height of the tree is $k$.

Next, we will prove (Theorem~\ref{Theorem_kernelIsomorphism}) that every element of $\mathbb{K}_{sp}$ is derived exclusively from a combination of an element of $s\mathbb{K}_p$ and an element of $p\mathbb{K}_s$. Let $\beta p\mathbb{K}_s =\{\beta px \pmod {sp}\mid x\in \mathbb{K}_s\}=\beta p\cdot\mathbb{K}_s$, which is a permutation of $p\mathbb{K}_s$, and likewise let $-\alpha s\mathbb{K}_p=-\alpha s.\mathbb{K}_p$, which is a permutation of $s\mathbb{K}_p$. Firstly, we will prove that the kernel $\mathbb{K}_{sp}$ of $\mathbb{Z}_{sp}$, is a multiplicative group. 

Recall that the primes $s$ and $p$ can be written as $s=2^nq+1$ and $p = 2^mr+1$, with $q$ and $r$ odd.

Recall also that for a map $f: \mathbb{Z}_{sp}^\ast \to \mathbb{Z}_{sp}^\ast$ defined by $f(x) = x^2 \bmod sp$, the kernel of the iterated map $f^n$ consists of all elements $x \in \mathbb{Z}_{sp}^\ast$ whose $n$-th iterate under squaring is congruent to $1$ modulo $sp$. This set forms a subgroup of the multiplicative group $\mathbb{Z}_{sp}^\ast$ and is denoted by $\mathbb{K}_{sp} = \ker(f^n)$.

\begin{theorem}
\label{Theorem_Ks_MultGroup}
$\mathbb{K}_{sp}$ is a multiplicative group. 
\end{theorem}

\begin{proof}
The kernel of the quadratic function $f$ over $\mathbb{Z}_{sp}$ is $\mathbb{K}_{sp}$ and it is defined as $\mathbb{K}_{sp}=Ker(f^n\mid_{\mathbb{Z}_{sp}})=\{x\mid x^{2^i}\equiv 1 \pmod{sp}, i=0,1,\ldots,n\}$. % where $i$ is the minimum non negative integer such that $x^{2^i}\equiv 1 \pmod{sp} \}$.
Let $x$ and $y$ be elements of $\mathbb{K}_{sp}$. Then, there exist $n_1$ and $n_2$ such that 
\begin{equation}
\label{eq_x_raisedTo_n_1}
x^{2^{n_1}}\equiv1\pmod {sp}
\end{equation}
and
\begin{equation}
\label{eq_x_raisedTo_n_2}
y^{2^{n_2}}\equiv1\pmod {sp}
\end{equation}
Obviously, this set does not contain any divisor of 0. Wlog we assume that $n_1 \ge n_2$. Raising both members of the congruence \ref{eq_x_raisedTo_n_2} to the power of $2^{n_1-n_2}$ gives
\begin{equation}
\label{eq_x_raisedTo_diff}
y^{2^{n_1}}\equiv1\pmod {sp}.
\end{equation}
The product of the left sides of congruences \ref{eq_x_raisedTo_n_1} and \ref{eq_x_raisedTo_diff} is also congruent to 1 (i.e., powers of $x$, $y$ are not divisors of zero) and is $(xy)^{2^{n_1}}\equiv1\pmod {sp}$, which proves that $\mathbb{K}_{sp}$ is closed under multiplication.

Furthermore, associativity obviously holds and the identity element is the unity (1). Next, we prove the existence of the inverse element. It is well known that all elements of $\mathbb{Z}_{sp}$ except multiples of $s$ and $p$ have an inverse, and it is not a multiple of either $s$ or $p$. Obviously, $\mathbb{K}_{sp}$ does not contain multiples of either $s$ or $p$ since any of their products gives zero. Then, let $x$ be an element of $\mathbb{K}_{sp}\subset\mathbb{Z}_{sp}$ and let $y$ be its inverse element. There are non-negative integers $n_1$ and $n_2$ (again, wlog we assume that $n_1\geq n_2$) such that $x^{2^{n_1}}\equiv1\pmod {sp}$ and $xy\equiv1\pmod {sp}$. Raising both members of the congruence to the power of $max(n_1,n_2)$ we get the following sequence of equivalent congruences $(xy)^{2^{n_1}}\equiv 1$, $x^{2^{n_1}}.y^{2^{n_1}} \equiv 1$, $1^{2^{(n_1-n_2)}}.y^{2^{n_1}}\equiv 1$, $y^{2^{n_1}}\equiv 1$, all of them $\pmod {sp}$. The later proves that $y=x^{{2^{n_1}}-1}$ must be in $\mathbb{K}_{sp}$. Thus, the existence of the inverse element has been proved.
\end{proof}

It is well known that the above defined kernels $s\mathbb{K}_p$ and $p\mathbb{K}_s$ are multiplicative groups. We consider the Cartesian product $s\mathbb{K}_p\times p\mathbb{K}_s$ of the kernels of $s\mathbb{F}_p$ and $p\mathbb{F}_s$. The next theorem proves that it is a multiplicative group.

\begin{theorem}
The set $s\mathbb{K}_p\times p\mathbb{K}_s \subset s\mathbb{F}_p\times p\mathbb{F}_s$ is a multiplicative group.
\end{theorem}

\begin{proof}
The sets $s\mathbb{K}p$ and $p\mathbb{K}s$ are closed under multiplication. Hence, it is easy to prove that their Cartesian product is also closed. Their unities form a pair consisting of their fixed points, i.e., $(-\alpha s, \beta p)$, which is the unity of the $s\mathbb{K}_p\times p\mathbb{K}_s$. The inverse of a pair of $s\mathbb{K}_p\times p\mathbb{K}_s$ is the pair of the inverses of $s\mathbb{K}_p$ and $p\mathbb{K}_s$, correspondingly.
\end{proof}

The two previous theorems paved the way for the following theorem that proves several results, such as that $\mathbb{K}_{sp}$ is a quad tree and its height is equal to the maximum of the heights of the corresponding trees of $\mathbb{K}_s$ and $\mathbb{K}_p$.

\begin{theorem}
\label{Theorem_kernelIsomorphism}
The multiplicative group $\mathbb{K}_{sp}$ is isomorphic to the Cartesian product $s\mathbb{K}_p\times p\mathbb{K}_s$.
\end{theorem}

\begin{proof}
Let $s=2^kq+1$ and $p=2^lr+1$ be prime numbers and also let $\alpha$, $\beta$ be the Euclidean coefficients such that $1=-\alpha s+\beta p$. The kernels of $f^n, f(x)=x^2$ for the sets $s\mathbb{F}_p$ and $p\mathbb{F}_s$ are $s\mathbb{K}_p$ and $p\mathbb{K}_s$, respectively. So, $s\mathbb{K}_p=\{x\mid x^{2^i}\equiv -\alpha s\pmod {sp}, i=1,2,\dots,l\}$ and $p\mathbb{K}_s=\{x\mid x^{2^i}\equiv \beta p\pmod {sp}, i=1,2,\dots,k\}$ where $\beta p$ and $-\alpha s$ are the fixed point elements of $\mathbb{Z}_{sp}$. Also, let $h$ be the function defined in Theorem~\ref{Theorem_mainIsomorphism_Zsp}, i.e., $h:\mathbb{K}_{sp}\mapsto s\mathbb{K}_p\times p\mathbb{K}_s$. Next, we prove that function $h$ is one-to-one and onto. For any $w=xs+yp\in \mathbb{Z}_{sp}$ it holds that $h(w)=h(xs+yp)= (hs,yp)\in s\mathbb{K}_p\times p\mathbb{K}_s$. For every $xs\in s\mathbb{K}_p$ $\exists \mu \in [1\dots l]$ such that $(xs)^{2^\mu}\equiv -\alpha s$ and for every $yp\in p\mathbb{K}_s$, $\exists \nu \in [1\dots k]$ such that $(yp)^{2^\nu}\equiv \beta p$. Assuming, wlog that $\mu\geq\nu$, it holds that $w^{2^{max(\mu,\nu)}}\equiv (xs)^{2^{\mu}} + (yp)^{2^{\mu}}\equiv -\alpha s+(\beta p)^{2^{\mu-\nu}}\!\equiv -\alpha s+\beta p \equiv 1 \pmod {sp}$. So, $s\mathbb{K}_p\times p\mathbb{K}_s \subseteq h(\mathbb{K}_{sp})$. Next, we prove that $h(\mathbb{K}_{sp}) \subseteq s\mathbb{K}_p\times p\mathbb{K}_s$. The kernel of $\mathbb{Z}_{sp}$ is $\mathbb{K}_{sp}=\{x\mid x^{2^i}\equiv 1 \pmod {sp}, i=0,1,2,\dots,\nu$, where $i$ is the minimum non negative integer with this property$\}$. For any $w\in \mathbb{K}_{sp}$, an integer $t\in [1\dots \nu]$ exists, such that $w^{2^t}\equiv 1$. Theorem~\ref{Theorem_mainIsomorphism_Zsp} proved that $w$ can be written as a unique combination of $s$ and $p$, i.e., $w=xs+yp$, where $xs\in s\mathbb{Z}_p$ and $yp\in p\mathbb{Z}_s$. This combined with $1=-\alpha s+\beta p$ gives that $(xs)^{2^t}+(yp)^{2^t}=-\alpha s+\beta p$. Thus, it holds that $(xs)^{2^t}=-\alpha s$ and $(yp)^{2^t}=\beta p$. Since $-\alpha s$ and $\beta p$ are the units (fixed points) of $s\mathbb{K}_p$ and $p\mathbb{K}_s$, respectively, it also holds that for every $w\in \mathbb{K}_{sp}$ there are $xs\in s\mathbb{K}_p$ and $yp\in p\mathbb{K}_s$ which gives $h(\mathbb{K}_{sp})\subseteq s\mathbb{K}_p\times p\mathbb{K}_s$. This, combined with the previously proved $s\mathbb{K}_p\times p\mathbb{K}_s \subseteq h(\mathbb{K}_{sp})$ proves that $h(\mathbb{K}_{sp})= s\mathbb{K}_p\times p\mathbb{K}_s$.
\end{proof}

\begin{table}[]
\begin{center}
\begin{tabular}{|l|c|c|c|c|c|}
\hline
                                                      & {$41\mathbb{K}_{29}$}  & \cellcolor[HTML]{C0C0C0}level 2                   & \cellcolor[HTML]{C0C0C0}level 2                    & \cellcolor[HTML]{EFEFEF}level 1                   & level 0                                           \\ \hline
$29\mathbb{K}_{41}$                                                  & $\mathbb{K}_{29\times 41}$ & 41 & 1148 & 492 & 697 \\ \hline
\rowcolor[HTML]{9B9B9B} 
\multicolumn{1}{|c|}{\cellcolor[HTML]{9B9B9B}level 3} & \cellcolor[HTML]{FFFFFF}232 & 273                                             & 191                                               & 724                                              & 929                                              \\ \hline
\rowcolor[HTML]{9B9B9B} 
level 3                                               & \cellcolor[HTML]{FFFFFF}957 & 998                                             & 916                                               & 260                                              & 465                                              \\ \hline
\rowcolor[HTML]{9B9B9B} 
level 3                                               & \cellcolor[HTML]{FFFFFF}290 & 331                                             & 249                                               & 782                                              & 987                                              \\ \hline
\rowcolor[HTML]{9B9B9B} 
level 3                                               & \cellcolor[HTML]{FFFFFF}899 & 940                                             & 858                                               & 202                                              & 407                                              \\ \hline
\rowcolor[HTML]{C0C0C0} 
\multicolumn{1}{|c|}{\cellcolor[HTML]{C0C0C0}level 2} & \cellcolor[HTML]{FFFFFF}319 & 360                                             & 278                                               & 811                                              & 1016                                             \\ \hline
\rowcolor[HTML]{C0C0C0} 
level 2                                               & \cellcolor[HTML]{FFFFFF}870 & 911                                             & 829                                               & 173                                              & 378                                              \\ \hline
\rowcolor[HTML]{EFEFEF} 
level 1                                               & \cellcolor[HTML]{FFFFFF}696 & \cellcolor[HTML]{C0C0C0}737                     & \cellcolor[HTML]{C0C0C0}655                       & 1188                                             & 204                                              \\ \hline
level 0                                               & \cellcolor[HTML]{FFFFFF}493 & \cellcolor[HTML]{C0C0C0}534                     & \cellcolor[HTML]{C0C0C0}452                       & \cellcolor[HTML]{EFEFEF}985                      & 1                                                \\ \hline
\end{tabular}
\end{center}
\caption{In this table the row header and the leftmost column are the elements of the trees ${29}\mathbb{K}_{41}$ and ${41}\mathbb{K}_{29}$, respectively. The inner values of the table are the elements of the quad tree corresponding to $\mathbb{K}_{29\times 41}$. When any two elements $xs$ and $yp$, located at levels $i$ and $j$ of $29\mathbb{K}_{41}$ and $41\mathbb{K}_{29}$ respectively, are added to form $xs+yp$, the resulting element $xs+yp$ is located at level $max\{i,j\}$. The table shows the levels of the resulting quad tree. Note that the intensity of grey in the background of each cell marks the level that the corresponding value belongs to (i.e., there are 16 level 3 values, 12 level 2 values, 3 level 1 values and 1 level 0 value, which is the root element).}
\label{Fig_levelsKsp}
\end{table}

\begin{figure}
\centering
\begin{subfigure}{0.06\textwidth}
    \includegraphics[width=\textwidth]{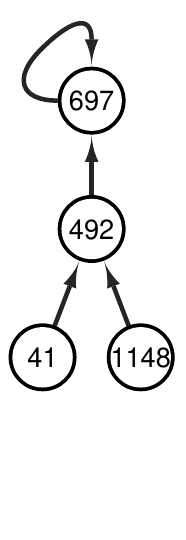}
    \caption{}
    \label{Fig_TreeKsp_NumExample_1}
\end{subfigure}
\hfill
\begin{subfigure}{0.13\textwidth}
    \includegraphics[width=\textwidth]{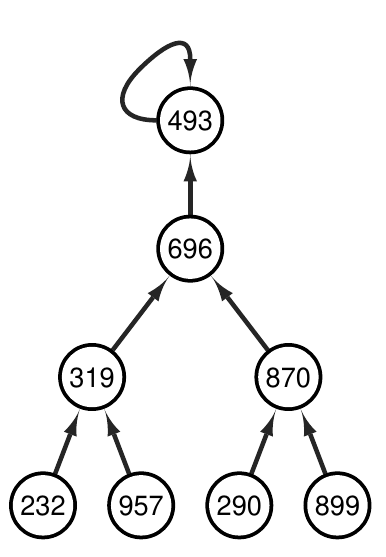}
    \caption{}
    \label{Fig_TreeKsp_NumExample_2}
\end{subfigure}
\hfill
\begin{subfigure}{0.75\textwidth}
    \includegraphics[width=\textwidth]{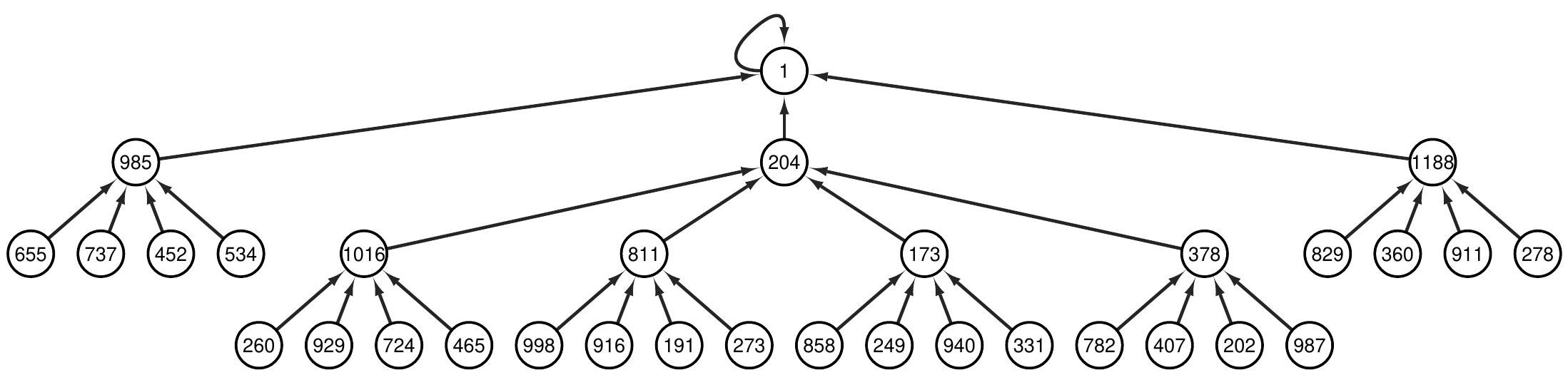}
    \caption{}
    \label{Fig_TreeKsp_NumExample_3}
\end{subfigure}
\caption{Subfigure~\ref{Fig_TreeKsp_NumExample_1} depicts the binary tree corresponding to ${41}\mathbb{K}_{29}$, while subfigure~\ref{Fig_TreeKsp_NumExample_2} depicts the binary tree corresponding to ${29}\mathbb{K}_{41}$. In~\ref{Fig_TreeKsp_NumExample_3}, the quad tree corresponding to $\mathbb{K}_{29\times 41}$, that results from the binary trees in~\ref{Fig_TreeKsp_NumExample_1} and~\ref{Fig_TreeKsp_NumExample_2}, is depicted. Note that these figures constitute the graphical representation of Table~\ref{Fig_levelsKsp}.}
\label{Fig_TreeKsp_NumExample}
\end{figure}
Corollary~\ref{corollary_treesResults} presents several results, directly extracted from the proof of Theorem~\ref{Theorem_kernelIsomorphism}. These results are also depicted visually in Figure~\ref{Fig_levelsKsp} and in Table ~\ref{Fig_levelsKsp}. %, \ref{Fig_TreeKsp_NumExample} and~\ref{Fig_TreeKsp}.

\begin{corollary}
\label{corollary_treesResults}
Let $s=2^kq+1$ and $p=2^lr+1$ be prime numbers, where $k$, $l$ are positive integers and $q$, $r$ are odd numbers. Based on Theorem~\ref{Theorem_kernelIsomorphism}, the following statements hold:
\begin{itemize}
\item
The graph $G=G_f(\mathbb{K}_{sp})$ is a quad tree and its height is $n=max\{k,l\}$.
\item %The graph $G$ of $f$ \hl{over} $\mathbb{K}_{sp}$ 
The graph $G$ is the tree $\mathcal{T}_n(1)$ and its size is $2^{k+l}$.
\item
Every non-leaf node of the tree of $G$, has exactly four pre-images (i.e., square roots).
\item
If $k=l$, then the tree $\mathcal{T}_n(1)$ of $G$ is a complete quad tree and its size is $2^{2k}=4^k$.
\item
For any $x_is \in s\mathbb{K}_p$ belonging to the $i^{th}$ level of the tree and any $y_jp\in p\mathbb{K}_s$ belonging to the  $j^{th}$ level of the tree, the element $w=x_is+y_jp$ of $G$ is located at the level $max\{i,j\}$ of $G$.
\end{itemize}
\end{corollary}

\subsection{\safeTitle{Cycles in $\mathbb{Z}_{sp}$}{Cycles in Zsp}}
\label{sec:sec3_2}

%A quadratic function $f$ in a ring is defined as $f:\mathbb{Z}_{sp}\mapsto \mathbb{Z}_{sp}$, where $f(x)=x^2\pmod {sp}$. 

Let $f^i$ be the composition of the quadratic function $f$ with itself $i$ times, for each positive integer $i$. This function will be used hereafter several times for the ring $\mathbb{Z}_{sp}$ and for its subgroups. The notation of cycles over finite fields that are defined below, are based on~\cite{panario2019functional} and~\cite{rogers1996graph}.

\begin{definition} \label{def_cycle}
Let $f$ be a quadratic function over a finite field $\mathbb{F}_s$ or the ring $\mathbb{Z}_{sp}$, with $s$ and $p$ prime. For $x$ in the respective domain, let $n$ be the smallest positive integer such that $f^n(x) = x$ (or $f^n(x) \equiv x \pmod{sp}$). The \emph{cycle} of $x$ under $f$ is 
$Cyc_n(x) = \{x, f(x), f^2(x), \ldots, f^{n-1}(x)\}$.
The set of all cycles of length $n$ is denoted by $Cyc_n(\mathbb{F}_s)$ or $Cyc_n(\mathbb{Z}_{sp})$, and the set of all cycles of any length by $Cyc(\mathbb{F}_s)$ or $Cyc(\mathbb{Z}_{sp})$. The \emph{order} of a cycle is its length, so for any $x$ we write $ord(x) = |Cyc_n(x)|$.
\end{definition}

The next theorem proves that two cycles $C_s \in Cyc(p\mathbb{F}_s)$ and $C_p \in Cyc(s\mathbb{F}_p)$, combined in $\mathbb{Z}_{sp}$, using function $h$, described in Theorem~\ref{Theorem_mainIsomorphism_Zsp}, produce a number of cycles of the same size. 

Note that for integers $x$ and $y$, $(x s + y p)^{2^n} \equiv (x s)^{2^n} + (y p)^{2^n} \pmod{sp}$.

\begin{lemma}
\label{Lemma_Cyclic}
Let $w=xs+yp$ be a cyclic element of $\mathbb{Z}_{sp}$. Then, the elements $xs$ and $yp$ are cyclic elements of $s\mathbb{F}_p$ and $p\mathbb{F}_s$, correspondingly.
\end{lemma}

\begin{proof}
There exists some integer $n$ for every cyclic element $w$ such that $w^{2^n}=w$. Thus, $(xs)^{2^n}+(yp)^{2^n}=xs+yp$, which gives $(xs)^{2^n}=xs$ and $(yp)^{2^n}=yp$. This means that both $xs$ and $yp$ are cyclic elements of $s\mathbb{F}_p$ and $p\mathbb{F}_s$, correspondingly.
\end{proof}

\begin{theorem}
\label{Theorem_Cycles_Zsp_Correspondence}
Let $s=2^kq+1$ and $p=2^lr+1$ be prime numbers (where $k$, $l$ are positive integers and $q$, $r$ are odd numbers), and let $f$ be a quadratic function over $\mathbb{Z}_{sp}$. Then every cycle $C_0 \in \mathrm{Cyc}_n(\mathbb{Z}_{sp})$ can be uniquely decomposed as $w = xs + yp$, where $xs \in C_s \in \mathrm{Cyc}_{\mu}(s\mathbb{F}_p)$ and $yp \in C_p \in \mathrm{Cyc}_{\nu}(p\mathbb{F}_s)$. Conversely, any pair of cycles $C_s$ and $C_p$ defines $\gcd(\mu,\nu)$ cycles in $\mathbb{Z}_{sp}$ of length $\mathrm{lcm}(\mu,\nu)$. 
\end{theorem}

\begin{proof}
Let $w = xs + yp \in C_0$. By the decomposition $\mathbb{Z}_{sp} = s\mathbb{F}_p \times p\mathbb{F}_s$, each iterate of $w$ under $f$ splits as 
\[
w^{2^i} = (xs)^{2^i} + (yp)^{2^i}, \quad i=0,1,\dots,n-1.
\]
Define $\mu$ and $\nu$ as the minimal positive integers such that $(xs)^{2^\mu} \equiv xs \pmod{sp}$ and $(yp)^{2^\nu} \equiv yp \pmod{sp}$. Then $C_s$ and $C_p$ have lengths $\mu$ and $\nu$, respectively, and $n = \mathrm{lcm}(\mu,\nu)$. Conversely, for any pair of cycles $C_s$ and $C_p$, the set 
\[
\{ xs + yp \mid xs \in C_s, \ yp \in C_p \}
\]
splits into $\gcd(\mu,\nu)$ cycles of length $\mathrm{lcm}(\mu,\nu)$.
\end{proof}

\begin{figure}
\centering
\begin{subfigure}{0.3\textwidth}
    \includegraphics[width=\textwidth]{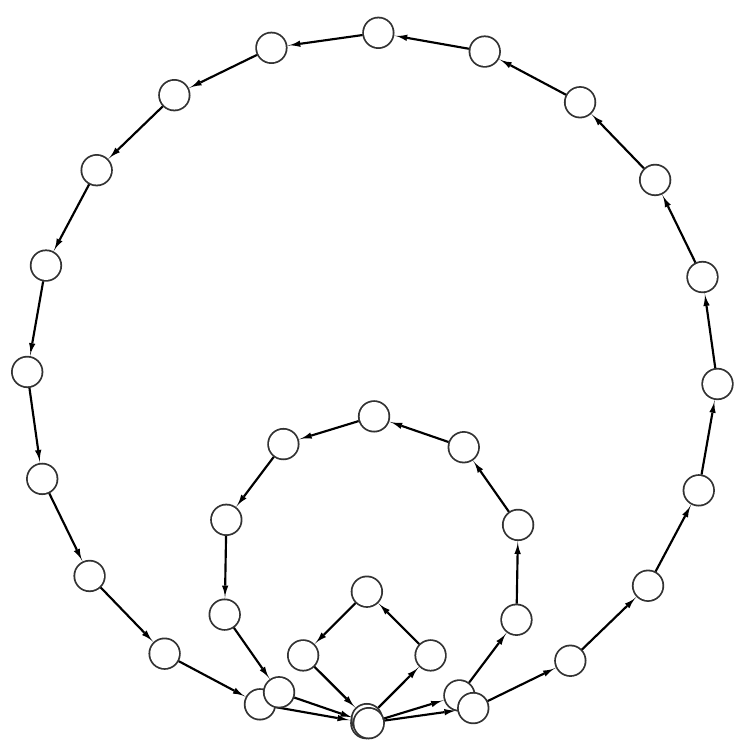}
    \caption{}
    \label{fig:cycles3_in_1_a}
\end{subfigure}
\hfill
\begin{subfigure}{0.3\textwidth}
    \includegraphics[width=\textwidth]{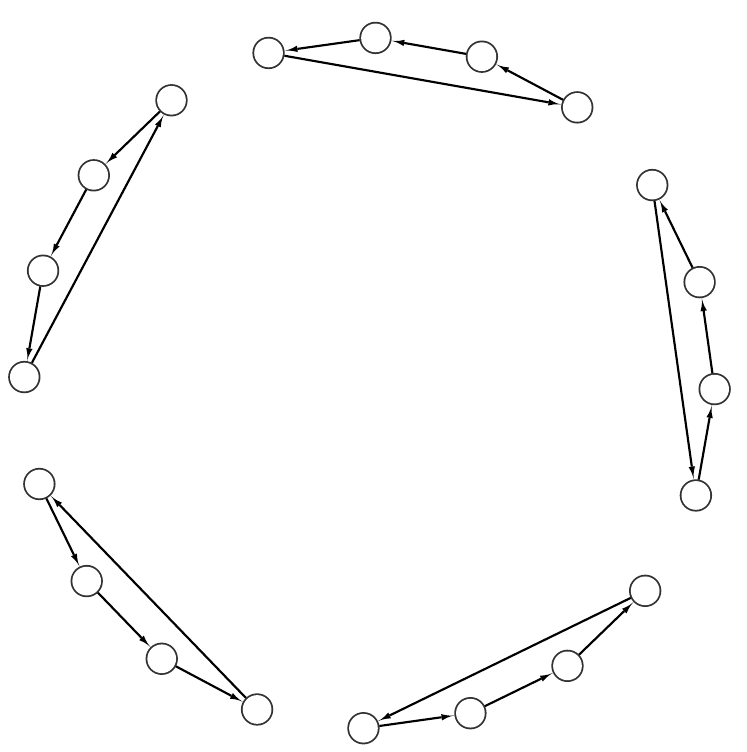}
    \caption{}
    \label{fig:cycles3_in_1_b}
\end{subfigure}
\hfill
\begin{subfigure}{0.3\textwidth}
    \includegraphics[width=\textwidth]{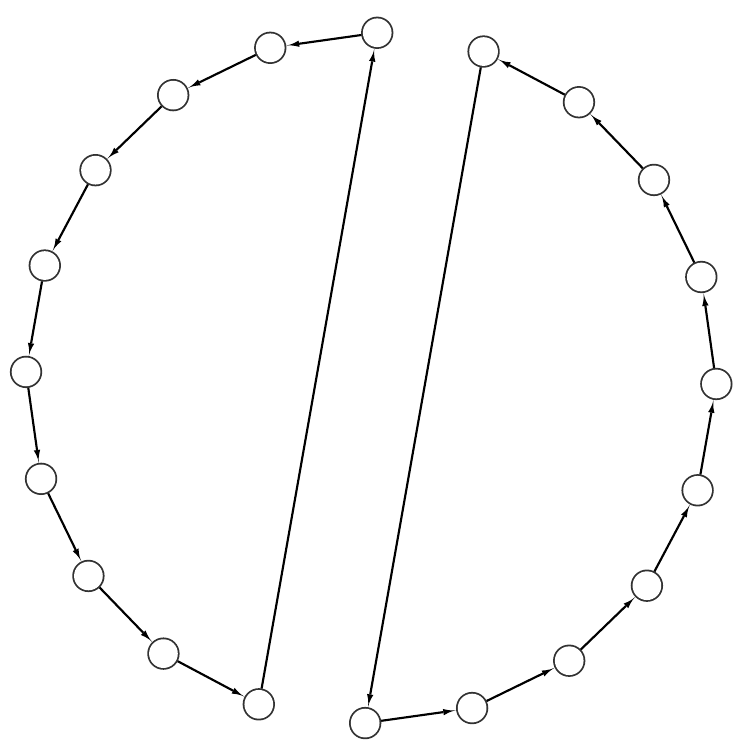}
    \caption{}
    \label{fig:cycles3_in_1_c}
\end{subfigure}
\caption{Subfigure \ref{fig:cycles3_in_1_a} shows the cycles $C_4$ of $23\mathbb{F}_{11}$ and $C_{10}$ of $11\mathbb{F}_{23}$, which based on Theorem~\ref{Theorem_Cycles_Zsp_Correspondence} create the external cycle $C_{20}$ of the ring $\mathbb{Z}_{11 \times 23}$. In \ref{fig:cycles3_in_1_b}, the 5 inner $C_4$ cycles of $C_{20}$ are depicted. In \ref{fig:cycles3_in_1_c}, the 2 inner $C_{10}$ cycles are shown.}
\label{fig:cycles3_in_1}
\end{figure}

\begin{remark}[Inner cycles of $C_0$]  
\label{Corollary_InnerCycles}  
Let $C_0 \subset \mathbb{G}_{sp}$ be a cycle of length $\theta$. Then there exist two inner cycles $C_s \subset s\mathbb{F}_p$ and $C_p \subset p\mathbb{F}_s$ of lengths $\mu$ and $\nu$, respectively, such that  
\[
\theta = \mathrm{lcm}(\mu,\nu).
\]  
Moreover, for every $w \in C_0$, the following congruences hold:  
\[
w^{2^\mu} \equiv w \pmod{s}, \qquad w^{2^\nu} \equiv w \pmod{p}.
\]  
\end{remark}  

\begin{remark}  
The congruences in remark~\ref{Corollary_InnerCycles} reveal structural properties of cycles in $\mathbb{G}_{sp}$. These properties underlie the cyclic attacks, and Rivest~\cite{rivest2001strong} describes techniques to protect cryptographic algorithms against such attacks.  
\end{remark}

\subsection{Arcs and Trees} \label{sec:sec3_3}
So far we have analysed the way that the tree of $\mathbb{K}_{sp}$ results from the trees of the sets $s\mathbb{K}_p$ and $p\mathbb{K}_s$. In section~\ref{sec:sec3_2}, we have shown how the cycles of $\mathbb{Z}_{sp}$ are generated by the cycles of $s\mathbb{F}_p$ and $p\mathbb{F}_s$. Here, we will describe and prove the existence of tree structures that are rooted at the vertices of the cycles. In order to do so, we define the concept of an arc. An arc, initially introduced in our previous work~\cite{verykios2024structures}, is defined to be a part of a cycle in any cyclic finite field. Likewise, we use the definition of trees for finite rings, also from~\cite{verykios2024structures}. 

\begin{definition}
\label{ArcTree}
Let $a$ be an element of the ring $\mathbb{Z}_{sp}$, where $s$ and $p$ are distinct odd primes. Let $\mathcal{A}_n(a)$ be an \emph{arc} consisting of $n+1$ nodes $(a_0, a_1, \ldots, a_n)$ such that
\[
a_{i-1} = f(a_i) = a_i^2, \quad i = 1, \ldots, n.
\]
Then each node can be expressed in terms of $a_n$ as
\[
a_i = f^{n-i}(a_n) = a_n^{2^{\,n-i}}, \quad i = 0,1,\ldots,n.
\]
By construction, each $a_i$ belongs to a cycle, so $\mathcal{A}_n(a)$ is indeed an arc. We assume that $a$ is the first node of the sequence, i.e., $a = a_0 = f^n(a_n) = a_n^{2^n}$.  

Moreover, if $a$ is any node of a cycle in a finite ring, we define the \emph{tree with root $a$} to be the set
\[
\mathcal{T}_n(a) = \{ x \mid f^i(x) = x^{2^i} = a \text{ for some } 0 \le i \le n \}.
\]
\end{definition}

\begin{remark} A \emph{tree} collects all elements whose iterates eventually reach a root node on a cycle, while an \emph{arc} is a segment of a cycle consisting of consecutive nodes, each of which serves as the root of a tree. \end{remark}
In the case of a finite field with a prime number of elements, the tree of a node can be visualized as a binary tree with root at the cyclic element $a$ (level 0), whereas in the case of a finite ring $\mathbb{Z}_{sp}$, the tree resembles a quadtree.  

Let $n = \max\{k,l\}$. Following \cite{verykios2024structures}, we define the set of tree structures of the cyclic elements as 
\[
T(\mathbb{Z}_{sp}) = \{\mathcal{T}_n(a) \mid a \in f^n(\mathbb{Z}_{sp})\}.
\]  
Similarly, let 
\[
A(\mathbb{Z}_{sp}) = \{\mathcal{A}_n(a) \mid a \in f^n(\mathbb{Z}_{sp})\}.
\]  
Then $A(\mathbb{Z}_{sp})$ and $T(\mathbb{Z}_{sp})$ are the sets of arcs and trees of $G(\mathbb{Z}_{sp})$, respectively.  

We have already proved that $\mathbb{K}_{sp}$ is represented graphically by the tree $\mathcal{T}_n(1)$. Henceforth, we denote this tree as 
\[
\mathcal{T}(\mathbb{K}_{sp}) = \mathcal{T}_n(1).
\]

The dynamics of the quadratic map $f(x) = x^2 \bmod sp$ on $\mathbb{Z}_{sp}$
naturally decompose into cycles (periodic orbits) and attached trees of preimages.
An \emph{arc} $\mathcal{A}_n(a)$ corresponds to a finite segment of a cycle
obtained by iterating $f$, while a \emph{tree} $\mathcal{T}_n(b)$
represents the backward structure of preimages rooted at $b$.
To combine these two dynamical components, we introduce
a levelwise multiplicative operation between an arc and a tree.
This operation, denoted by $\bigodot$,
acts by multiplying corresponding levels of the arc and the tree
in $\mathbb{Z}_{sp}$, yielding a new tree whose nodes inherit
the same functional relations under $f$.
The formal definition follows.

\begin{definition}[Multiplication of an Arc and a Tree]
\label{def:ArcByTree}
Let $\mathcal{A}_n(a) = (a_0, a_1, \ldots, a_n)$ be an \emph{arc} in $\mathbb{Z}_{sp}$,
generated by the iteration of the quadratic map $f(x) = x^2 \bmod sp$, such that
\[
a_i = f^{\,n-i}(a_n) = a_n^{2^{\,n-i}}, \quad i = 0,1,\ldots,n,
\]
and $a_0 = f^{\,n}(a_n) = a_n^{2^n} = a$.
Each $a_i$ thus lies on a cycle of $f$.

Let $\mathcal{T}_n(b)$ be a \emph{tree} in $T(\mathbb{Z}_{sp})$, with nodes denoted by $b_{ij}$,
where $i = 0,1,\ldots,n$ indicates the level ($i=0$ is the root, $i=n$ the leaves)
and each node has at most four preimages under $f$.
We assume the standard binary indexing convention:
\[
f(b_{ij}) = b_{i-1,m}, \qquad
b_{ij} = f^{\,n-i}(b_{n\nu}) = b_{n\nu}^{2^{\,n-i}}.
\]
where \(m\) is the unique index of the \emph{parent} of \(b_{ij}\) at level \(i-1\), and \(\nu\) is the unique index of a leaf in the descendant set of \(b_{ij}\).

Define the mapping
\[
\bigodot : A(\mathbb{Z}_{sp}) \times T(\mathbb{Z}_{sp}) \longrightarrow T(\mathbb{Z}_{sp})
\]
by
\[
(\mathcal{A}_n(a), \mathcal{T}_n(b)) \longmapsto \mathcal{T}_n(c),
\]
where the nodes of level $i$ in $\mathcal{T}_n(c)$ are given by
\[
c_{ij} = a_i b_{ij} \bmod sp, \qquad
j = 1 \text{ when } i = 0, \quad j = 1,\ldots,2^{i-1} \text{ when } i = 1,\ldots,n.
\]
\end{definition}

\begin{theorem}
\label{theorem:ArcByTree}
Assume that $\gcd(a_n, sp) = 1$, so that each $a_i$ is a unit in $\mathbb{Z}_{sp}$.
Then the mapping
\[
\bigodot : A(\mathbb{Z}_{sp}) \times T(\mathbb{Z}_{sp}) \to T(\mathbb{Z}_{sp})
\]
defined in Definition~\ref{def:ArcByTree} is well-defined.
Moreover, the root of the resulting tree $\mathcal{T}_n(c)$ satisfies
\[
c = a b.
\]
\end{theorem}

\begin{proof}
Let $\mathcal{A}_n(a)$ and $\mathcal{T}_n(b)$ be as above, and define
\[
c_{n\nu} = a_n b_{n\nu} \bmod sp
\]
for each leaf $\nu$ at level $n$.
For $i < n$, define
\[
c_{ij} = f^{\,n-i}(c_{n\nu}) = (c_{n\nu})^{2^{\,n-i}},
\]
where $\nu$ is the index of the descendant leaf of $b_{ij}$.
Since $\mathbb{Z}_{sp}$ is commutative, exponentiation distributes over multiplication:
\[
(c_{n\nu})^{2^{\,n-i}}
  = (a_n b_{n\nu})^{2^{\,n-i}}
  = a_n^{2^{\,n-i}} b_{n\nu}^{2^{\,n-i}}
  = a_i b_{ij}.
\]
Hence $c_{ij} = a_i b_{ij}$ for all $i,j$.

To verify structural consistency, let $b_{i-1,m}$ be the parent of $b_{ij}$.
Then
\[
f(c_{ij}) = (a_i b_{ij})^2 = a_i^2 b_{ij}^2 = a_{i-1} b_{i-1,m} = c_{i-1,m},
\]
so the parent–child relation is preserved.
Since each $a_i$ is a unit modulo $sp$, the map
$x \mapsto a_i x$ is bijective on $\mathbb{Z}_{sp}$,
and therefore the number and structure of preimages under $f$ are maintained.
Consequently, $\mathcal{T}_n(c)$ is a valid element of $T(\mathbb{Z}_{sp})$.

Finally, for the root we have $c_0 = a_0 b_0 = ab$, completing the proof.
\end{proof}

Let $a$ be a cyclic element of $\mathbb{G}_{sp}$, and let $\mathcal{A}_n(a)$ be an arc in $\mathbb{G}_{sp}$ of length $n = \max(k,l)$.  
Applying Theorem~\ref{theorem:ArcByTree} yields the following corollary:

\begin{corollary}[Arc–Tree Rooting]
\label{cor:ArcTreeRoot}
For any cyclic element $a \in \mathbb{G}_{sp}$, the quad tree $\mathcal{T}_n(1)$, when multiplied by the arc $\mathcal{A}_n(a)$ via the operation $\bigodot$, produces a tree rooted at $a$:
\[
\mathcal{T}_n(a) = \mathcal{A}_n(a) \bigodot \mathcal{T}_n(1).
\]
\end{corollary}

Consequently, each cyclic element $a$ of $\mathbb{G}_{sp}$ is associated with a tree 
$\mathcal{T}_n(a)$ that is structurally identical to the kernel tree $\mathcal{T}_n(1)$, 
with the only difference being the root node. This shows that the set of trees
$\{\mathcal{T}_n(a) \mid a \text{ cyclic in } \mathbb{G}_{sp}\}$ 
forms a family of isomorphic quad trees rooted at each cyclic element.

\subsection{\safeTitle{Isomorphism of $s\mathbb{K}_p \times p\mathbb{F}_s$}{Isomorphism of sKp x pFs}}
\label{sec:sec3_4}

In this subsection, we identify important properties of the structure of $\mathbb{Z}_{sp}$ 
related to the following sets: 
\begin{itemize}
\item $p\mathbb{F}_s^{e} = \{ y p + e \mid y \in [0, \dots, s-1],\ e \in s\mathbb{K}_p \}$,
\item $s\mathbb{F}_p^{e} = \{ x s + e \mid x \in [0, \dots, p-1],\ e \in p\mathbb{K}_s \}$.
\end{itemize}

The next theorem shows that the set $p\mathbb{F}_s^e$ forms a multiplicative group.

\begin{theorem}
The subset 
\[
p\mathbb{F}_s^{e} = \{ y p + e \mid y \in [0, \dots, s-1],\ e \in s\mathbb{K}_p \} \subset \mathbb{Z}_{sp}
\]
is a multiplicative group.
\end{theorem}

\begin{proof}
We verify the group properties that are not trivial: closure under multiplication, existence of the identity element, and existence of inverses for all non-zero elements.  

Let $x_1 p + e_1 s$ and $x_2 p + e_2 s$ be elements of $p\mathbb{F}_s^{e}$, 
where $x_1, x_2 \in [0, \dots, s-1]$ and $e_1 s, e_2 s \in s\mathbb{K}_p$.
Suppose that for some $i$ and $j$, $(e_1 s)^{2^i} \equiv -\alpha s \pmod{sp}$ and $(e_2 s)^{2^j} \equiv -\alpha s \pmod{sp}$.  
Then
\[
(x_1 p + e_1 s)(x_2 p + e_2 s) = x_1 x_2 p^2 + (e_1 x_2 + e_2 x_1) sp + e_1 e_2 s^2 \equiv x_1 x_2 p^2 + e_1 e_2 s^2 \pmod{sp}.
\]

By Theorem~\ref{TheoremFour:pFs_field}, $s\mathbb{F}_p$ is closed under multiplication, 
so $x_1 x_2 p^2 \in s\mathbb{F}_p$.  
Let $e = e_1 e_2 s \pmod p$; since $s\mathbb{K}_p$ is a group (Theorem~\ref{Theorem_kernelIsomorphism}), we have $e s \in s\mathbb{K}_p$.  
Thus, $p\mathbb{F}_s^e$ is closed under multiplication.

The multiplicative identity is clearly $\beta p - \alpha s = 1 \in p\mathbb{F}_s^e$.

For inverses, let $x p + e_1 s \in p\mathbb{F}_s^e$, and let its inverse be $y p + e_2 s$.  
Then
\[
(x p + e_1 s)(y p + e_2 s) = 1 = \beta p - \alpha s.
\]
Expanding gives $x y p^2 + e_1 e_2 s^2 \equiv \beta p - \alpha s$, 
which leads to the congruences
\[
x y p \equiv \beta \pmod s, \qquad e_1 e_2 s \equiv -\alpha \pmod p.
\]
Hence the inverse exists and can be computed as
\[
y \equiv \beta (x p)^{-1} \pmod s, \qquad e_2 \equiv -\alpha (e_1 s)^{-1} \pmod p.
\]

Therefore, $p\mathbb{F}_s^e$ satisfies all group axioms and is a multiplicative group.
\end{proof}

In a similar manner, it can be shown that 
\[
s\mathbb{F}_p^{e} = \{ x s + e \mid x \in [0, \dots, p-1],\ e \in p\mathbb{K}_s \}
\]
is a multiplicative group, provided that $p\mathbb{K}_s$ is the kernel of $f^n$ over $p\mathbb{F}_s$. \\

The next theorem establishes that $s\mathbb{K}_p \times p\mathbb{F}_s$ is a multiplicative group.

\begin{theorem}
The subset $s\mathbb{K}_p \times p\mathbb{F}_s$ of $s\mathbb{F}_p \times p\mathbb{F}_s$ 
is a multiplicative group under componentwise multiplication:
\[
(a_1,b_1)\cdot(a_2,b_2) = (a_1 a_2, b_1 b_2).
\]
\end{theorem}

\begin{proof}
The set $s\mathbb{K}_p \subset s\mathbb{F}_p$ is a multiplicative group, 
and both $s\mathbb{F}_p$ and $p\mathbb{F}_s$ are finite fields.  
Componentwise multiplication clearly satisfies closure and associativity.

The identity element is 
\[
(1_{s\mathbb{K}_p}, 1_{p\mathbb{F}_s}) = (-\alpha s, \beta p),
\]
formed by the identities of $s\mathbb{K}_p$ and $p\mathbb{F}_s$.  
Indeed, for any $(\alpha x s, \beta y p) \in s\mathbb{K}_p \times p\mathbb{F}_s$, we have
\[
(\alpha x s, \beta y p) \cdot (-\alpha s, \beta p) 
= (\alpha x s \cdot -\alpha s, \beta y p \cdot \beta p) 
= (\alpha x s, \beta y p),
\]
so the identity is correct.

For inverses, let $(\alpha x_1 s, \beta y_1 p)$ have inverse $(\alpha x_2 s, \beta y_2 p)$.  
Then we require
\[
(\alpha x_1 s, \beta y_1 p) \cdot (\alpha x_2 s, \beta y_2 p) = (-\alpha s, \beta p),
\]
which is equivalent to the two congruences:
\[
\alpha x_1 s \cdot \alpha x_2 s \equiv -\alpha s \pmod s, \qquad
\beta y_1 p \cdot \beta y_2 p \equiv \beta p \pmod p.
\]
Hence the inverses can be computed as
\[
\alpha x_2 s \equiv (-\alpha s) \cdot (\alpha x_1 s)^{-1} \pmod s, \qquad
\beta y_2 p \equiv (\beta p) \cdot (\beta y_1 p)^{-1} \pmod p.
\]
Therefore, every element has an inverse, and $s\mathbb{K}_p \times p\mathbb{F}_s$ is a multiplicative group.
\end{proof}

The next theorem establishes that the groups $s\mathbb{K}_p \times p\mathbb{F}_s$ and $p\mathbb{F}_s^{e}$ are isomorphic.

\begin{theorem}
The group $s\mathbb{K}_p \times p\mathbb{F}_s$ is isomorphic to $p\mathbb{F}_s^e$.
\end{theorem}
 
\begin{proof}
In Theorems~\ref{Theorem_mainIsomorphism_Zsp} and \ref{Theorem_kernelIsomorphism}, we defined a function
\[
h : \mathbb{Z}_{sp} \longrightarrow s\mathbb{F}_p \times p\mathbb{F}_s.
\]
Here, we restrict $h$ to the domain $p\mathbb{F}_s^e$:
\[
h : p\mathbb{F}_s^e \longrightarrow s\mathbb{K}_p \times p\mathbb{F}_s.
\]

Every element of $\mathbb{Z}_{sp}$ can be written uniquely as a linear combination of $s$ and $p$, as shown in Theorem~\ref{Theorem_mainIsomorphism_Zsp}.  
For $x s + y p \in p\mathbb{F}_s^e$, we can naturally separate it into two components, $x s$ and $y p$.  
By construction, $y p \in p\mathbb{F}_s$ and $x s \in s\mathbb{K}_p$, since $(x s)^{2^i} = -\alpha s$ for some $i$.  
Thus each element of $p\mathbb{F}_s^e$ corresponds uniquely to an element of $s\mathbb{K}_p \times p\mathbb{F}_s$, proving that $h$ is one-to-one.

To prove that $h$ is onto, observe that every $(x s, y p) \in s\mathbb{K}_p \times p\mathbb{F}_s$ is the image of $x s + y p \in p\mathbb{F}_s^e$.  

Finally, we verify that $h$ preserves multiplication.  
For $(x_1 s + y_1 p), (x_2 s + y_2 p) \in p\mathbb{F}_s^e$, we have
\[
h\big((x_1 s + y_1 p)(x_2 s + y_2 p)\big) = h(x_1 x_2 s^2 + y_1 y_2 p^2) = (x_1 x_2 s^2, y_1 y_2 p^2),
\]
while
\[
h(x_1 s + y_1 p) \cdot h(x_2 s + y_2 p) = (x_1 s, y_1 p) \cdot (x_2 s, y_2 p) = (x_1 x_2 s^2, y_1 y_2 p^2).
\]
Hence $h$ is a multiplicative homomorphism, and thus an isomorphism.
\end{proof}

\begin{remark}
The isomorphism $h : p\mathbb{F}_s^e \longrightarrow s\mathbb{K}_p \times p\mathbb{F}_s$ shows that every element of $p\mathbb{F}_s^e$ can be uniquely decomposed into a component in the kernel $s\mathbb{K}_p$ and a component in the field $p\mathbb{F}_s$.  
In other words, $p\mathbb{F}_s^e$ can be viewed as a structured combination of its kernel and field factors, which explains both its group properties and its compatibility with the product structure in $s\mathbb{K}_p \times p\mathbb{F}_s$.
\end{remark}

Similarly, in a symmetrical manner, it can be shown that $s\mathbb{F}_p \times p\mathbb{K}_s$ and $s\mathbb{F}_p^e$ are multiplicative groups that are isomorphic to each other.

\begin{figure}
\centering
\begin{subfigure}{0.12\textwidth}
    \includegraphics[width=\textwidth]{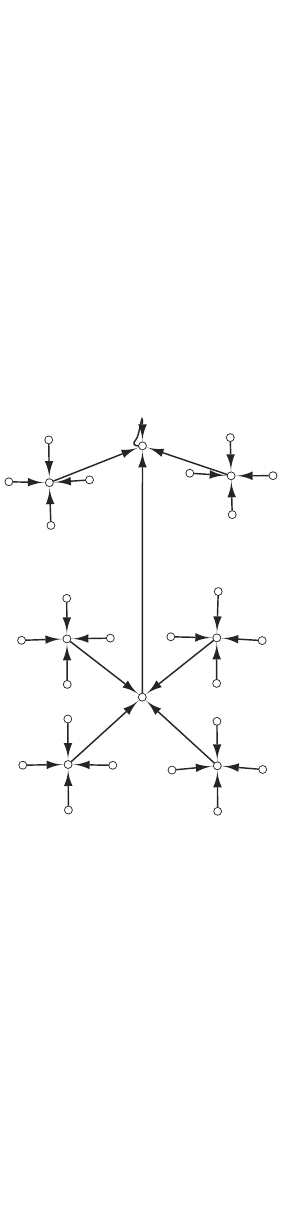}
    \caption{}
    \label{Fig_TreeKsp_1}
\end{subfigure}
\hfill
\begin{subfigure}{0.38\textwidth}
    \includegraphics[width=\textwidth]{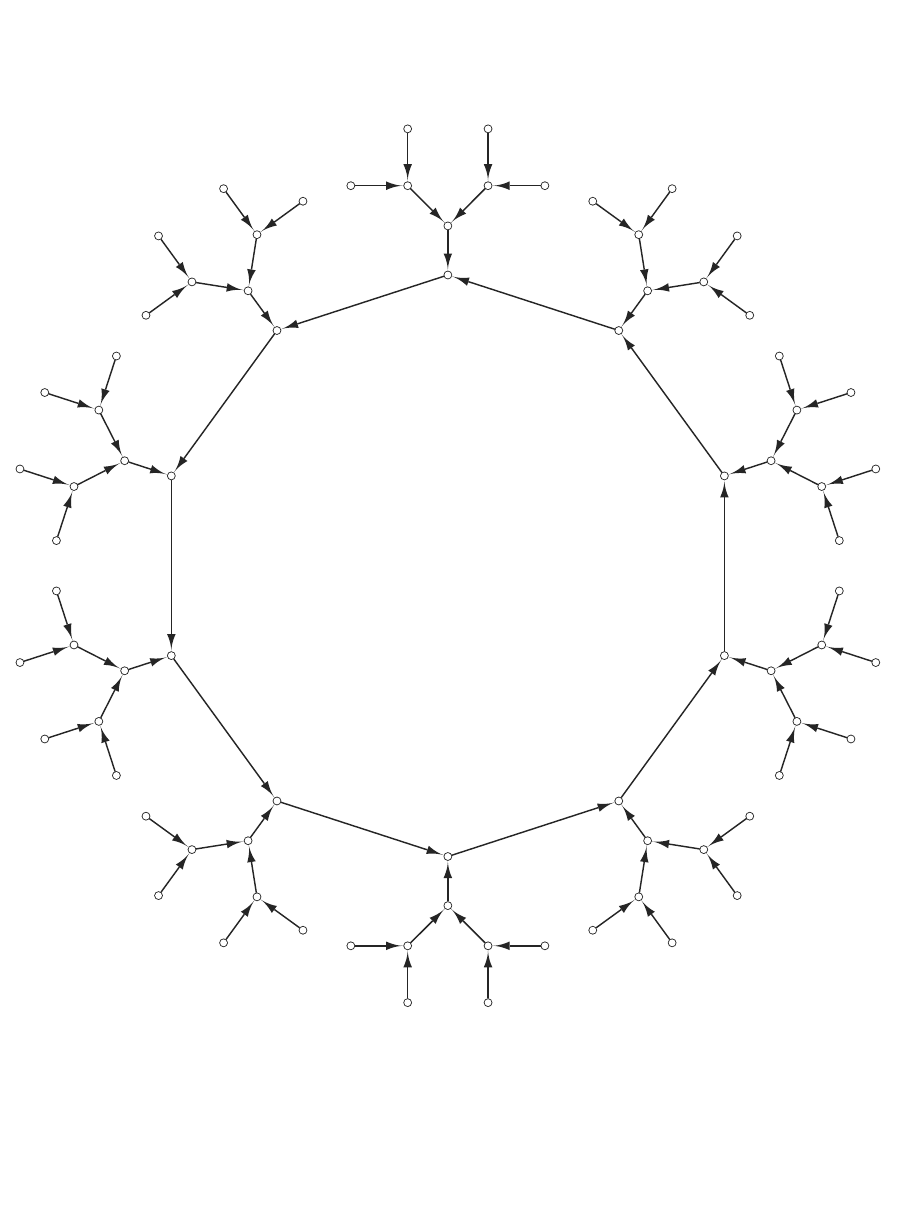}
    \caption{}
    \label{Fig_TreeKsp_2}
\end{subfigure}
\hfill
\begin{subfigure}{0.45\textwidth}
    \includegraphics[width=\textwidth]{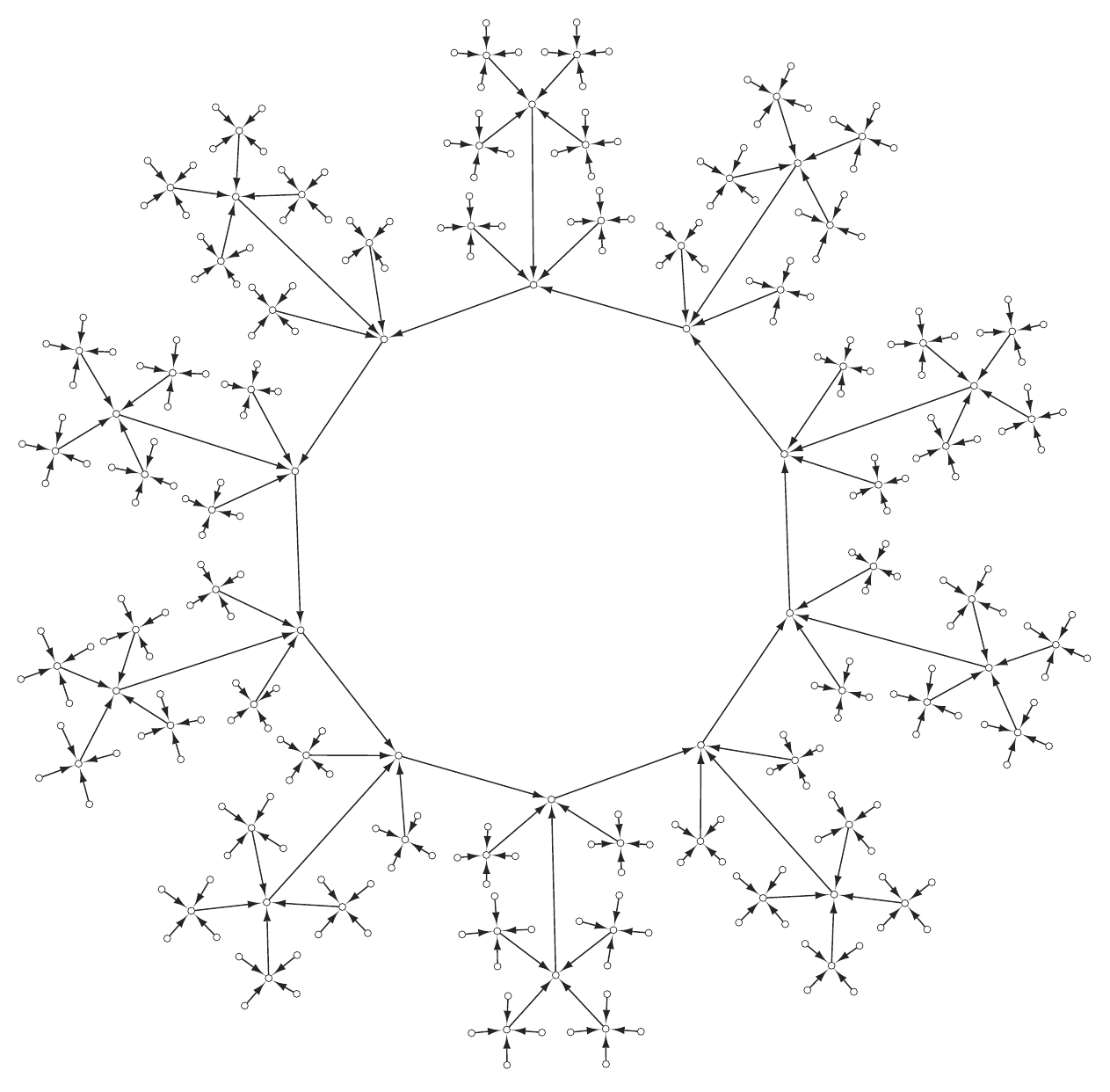}
    \caption{}
    \label{Fig_TreeKsp_3}
\end{subfigure}
\caption{In \ref{Fig_TreeKsp_1}, the quad tree representing the kernel is $\mathbb{F}_{29\times89}$. In \ref{Fig_TreeKsp_2} there is a typical cycle of $29\mathbb{F}_{89}$ which is isomorphic to $\mathbb{F}_{89}$. The trees rooted at its cyclic nodes are the binary trees isomorphic to $\mathbb{K}_{89}$. In \ref{Fig_TreeKsp_3}, the combined set $h(29\mathbb{F}_{89}^e)=(29\mathbb{F}_{89},89\mathbb{K}_{29})$, consisting of the cycle $29\mathbb{F}_{89}$, where each cyclic element is rooted by a tree isomorphic to the kernel tree shown in \ref{Fig_TreeKsp_1}. For each cyclic element $w$ of \ref{Fig_TreeKsp_3} it holds that $w=29t+1$.} 
\label{sFpxpKs}
\end{figure}

The elements of the two sets $s\mathbb{F}_p^e$ and $p\mathbb{F}_s^e$ satisfy the following special property.  
For any $w \in p\mathbb{F}_s^e$, we have
\[
w^{2^k} = y p + 1,
\]
and similarly, for any $z \in s\mathbb{F}_p^e$, 
\[
z^{2^l} = x s + 1.
\]
In other words, each element of these sets, when raised to a suitable power of two (with exponent at most $\max(k,l)$), yields a number of the form $x s + 1$ or $y p + 1$.  

We call all elements of $s\mathbb{F}_p^e \cup p\mathbb{F}_s^e$ \textbf{off-by-one multiples} of either $s$ or $p$.  
Consequently, any number that is not a multiple of $s$ or $p$ and does not belong to $s\mathbb{F}_p^e \cup p\mathbb{F}_s^e$ will have a distance of at least two from any power-of-two multiple of $s$ or $p$, regardless of how many times it is squared.

\section{\safeTitle{The set of cryptography $\mathbb{D}_{sp}$}{The set of cryptography Dsp}}
\label{sec:sec4}

In this section, we define the new set $\mathbb{D}_{sp}$ and describe its properties. This set will be of importance for cryptography, as will be explained below.

\begin{definition} \label{Definition_Delta_set}
The set $\mathbb{D}_{sp}$ is defined as
\begin{equation}
\label{DeltaSet}
\mathbb{D}_{sp} = \mathbb{G}_{sp} - s\mathbb{F}_p^e - p\mathbb{F}_s^e = \mathbb{Z}_{sp} - s\mathbb{F}_p - p\mathbb{F}_s - s\mathbb{F}_p^e - p\mathbb{F}_s^e.
\end{equation}

We also define the sets
\begin{equation}
\label{sFp_Clean}
s\mathbb{F}_p^{\ast\ast} = s\mathbb{F}_p - s\mathbb{K}_p - \{0\}, 
\qquad
p\mathbb{F}_s^{\ast\ast} = p\mathbb{F}_s - p\mathbb{K}_s - \{0\}.
\end{equation}
\end{definition}

We use the cyclic graph notation $(\mathrm{Cyc}_n(\mathbb{D}_{sp}), \mathcal{T}_n(\mathbb{K}_{sp}))$ introduced by Panario and Reis~\cite{panario2019functional}.  
Here, $\mathrm{Cyc}_n(\mathbb{D}_{sp})$ is a cyclic graph of length $n$ constructed from $\mathbb{D}_{sp}$, with a rooted tree $\mathcal{T}_n(\mathbb{K}_{sp})$ of height $n$ attached at each vertex. The structure satisfies:
\begin{itemize}
\item $\mathbb{K}_{sp}$ is the kernel (see Section~\ref{sec:sec3_2}),
\item $n = \max\{k,l\}$ is the height of the tree, and
\item each cyclic node $a \in \mathbb{D}_{sp}$ is the root of a tree $\mathcal{T}_n(a)$, which is isomorphic to $\mathcal{T}_n(1)$ derived from the kernel.
\end{itemize}

This construction results in a hybrid structure where the cycle $\mathrm{Cyc}_n(\mathbb{Z}_{sp})$ represents periodic repetition, and the attached trees represent pre-periodic elements associated with the kernel.  
All cyclic elements are pointed by trees isomorphic to $\mathbb{K}_{sp}$, as described in Theorem~\ref{Theorem_kernelIsomorphism}.

Based on previous sections, the function $h$ maps the sets as follows:
\begin{align*}
h(s\mathbb{F}_p) &= s\mathbb{F}_p \times \{0\}, &
h(p\mathbb{F}_s) &= \{0\} \times p\mathbb{F}_s, \\
h(s\mathbb{F}_p^e) &= s\mathbb{F}_p^* \times p\mathbb{K}_s, &
h(p\mathbb{F}_s^e) &= s\mathbb{K}_p \times p\mathbb{F}_s^*.
\end{align*}

Hence, the image of $\mathbb{D}_{sp}$ under $h$ is
\[
h(\mathbb{D}_{sp}) = (s\mathbb{F}_p - s\mathbb{K}_p - \{0\}) \times (p\mathbb{F}_s - p\mathbb{K}_s - \{0\}) = s\mathbb{F}_p^{\ast\ast} \times p\mathbb{F}_s^{\ast\ast}.
\]

The sets $s\mathbb{F}_p^{\ast\ast}$ and $p\mathbb{F}_s^{\ast\ast}$ consist of elements whose iterates under $f^i$ are neither the unity of their respective fields nor zero.  
Their Cartesian product forms $\mathbb{D}_{sp}$, which is comprised of trees rooted in cyclic elements.

\medskip

Each element $w \in \mathbb{Z}_{sp}$ falls into one of three categories:
\begin{enumerate}
\item \textbf{Multiples of $s$ or $p$:} $w \in s\mathbb{F}_p$ or $w \in p\mathbb{F}_s$.
\item \textbf{Off-by-one multiples:} $w \in s\mathbb{F}_p^e$ or $w \in p\mathbb{F}_s^e$, so that there exists $i \le \max\{k,l\}$ with $w^{2^i} \equiv 1 \pmod{s}$ or $w^{2^i} \equiv 1 \pmod{p}$. The intersection of these sets forms the kernel $\mathbb{K}_{sp}$, satisfying $w^{2^i} \equiv 1 \pmod{sp}$.
\item \textbf{Remaining elements ($\mathbb{D}_{sp}$):} For $w \in \mathbb{D}_{sp}$, all powers $w^{2^i}$ differ from multiples of $s$ or $p$ by at least 2. Each such element either lies on a tree directed toward a cyclic element or is itself cyclic, with the eventual cycle reached at $w^{2^{\max\{k,l\}}}$.
\end{enumerate}

\medskip

Let $s = 2^k q + 1$ and $p = 2^l r + 1$, with $q$ and $r$ odd.  

\begin{itemize}
\item \textbf{Category 1:} $s\mathbb{F}_p$ and $p\mathbb{F}_s$ have $p$ and $s$ elements respectively, with 0 as the only common element, giving $s + p - 1$ elements.
\item \textbf{Category 2:} $|s\mathbb{F}_p^e \cup p\mathbb{F}_s^e| = 2^l(s-1) + 2^k(p-1) - 2^{k+l} = 2^{k+l}(q+r-1)$, after removing overlaps and zero.
\item \textbf{Category 3:} $\mathbb{D}_{sp}$ has
\[
|\mathbb{D}_{sp}| = sp - 2^{k+l}(q+r-1) - s - p + 1 = 2^{k+l}(q-1)(r-1)
\]
elements. The number of cyclic elements in $\mathbb{D}_{sp}$ is $(q-1)(r-1)$, and the remaining elements form trees directed toward these cyclic elements.
\end{itemize}

Following Rivest and Silveman~\cite{rivest2001strong}, define $s^- = s-1 = 2^k q$ and $q^{--}$ as the largest prime factor of $q-1$. Similarly, $p^- = p-1 = 2^l r$ and $r^{--}$ as the largest prime factor of $r-1$.  
Then the maximum cycle length of $G(p\mathbb{F}_s)$ is $q^{--}$, and for $G(s\mathbb{F}_p)$ it is $r^{--}$.  
Within $\mathbb{D}_{sp}$, the maximum cycles have length $\mathrm{lcm}(q^{--}, r^{--})$, containing inner cycles of lengths $q^{--}$ and $r^{--}$.  
See Figure~\ref{fig:cycles3_in_1} for an example with $s = 11$ and $p = 23$.

\section{\safeTitle{Cycles, Trees, and the Cryptographic Set $\mathbb{D}_{sp}$}{Cycles, Trees, and the Cryptographic Set Dsp}}
\label{sec:sec5}

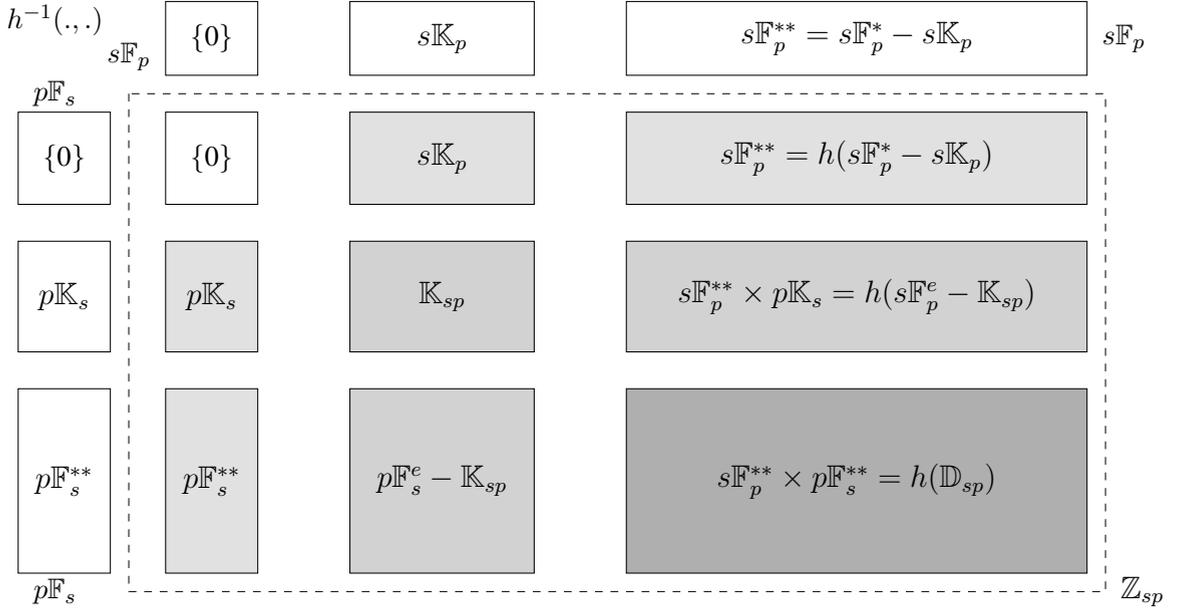
\begin{figure}
\centering
\resizebox{1\textwidth}{!}{%
\begin{tikzpicture}
\tikzstyle{every node}=[font=\normalsize]
\draw  (3,11.25) rectangle  node {\normalsize \{0\}} (4.25,10);
\draw  (5,12.75) rectangle  node {\normalsize \{0\}} (6.25,11.75);
\draw  (3,9.5) rectangle  node {\large $p\mathbb{K}_s$} (4.25,8);
\draw  (7.5,12.75) rectangle  node {\large $s\mathbb{K}_p$} (10,11.75);
\draw  (11.25,12.75) rectangle  node {\large $s\mathbb{F}^{**}_p=s\mathbb{F}^*_p-s\mathbb{K}_p$} (17.5,11.75);
\draw  (3,7.5) rectangle  node {\large $p\mathbb{F}^{**}_s$} (4.25,5);
\draw  (5,11.25) rectangle  node {\normalsize \{0\}} (6.25,10);
\draw [ fill={rgb,255:red,225; green,225; blue,225} ] (7.5,11.25) rectangle  node {\large $s\mathbb{K}_p$} (10,10);
\draw [ fill={rgb,255:red,225; green,225; blue,225} ] (5,9.5) rectangle  node {\large $p\mathbb{K}_s$} (6.25,8);
\draw [ fill={rgb,255:red,210; green,210; blue,210} ] (7.5,9.5) rectangle  node {\large $\mathbb{K}_{sp}$} (10,8);
\draw [ fill={rgb,255:red,225; green,225; blue,225} ] (11.25,11.25) rectangle  node {\large $s\mathbb{F}^{**}_p=h(s\mathbb{F}^*_p-s\mathbb{K}_p)$} (17.5,10);
\draw [ fill={rgb,255:red,225; green,225; blue,225} ] (5,7.5) rectangle  node {\large $p\mathbb{F}^{**}_s$} (6.25,5);
\draw [ fill={rgb,255:red,210; green,210; blue,210} ] (7.5,7.5) rectangle  node {\large $p\mathbb{F}_s^e-\mathbb{K}_{sp}$} (10,5);
\draw [ fill={rgb,255:red,210; green,210; blue,210} ] (11.25,9.5) rectangle  node {\large $s\mathbb{F}^{**}_p\times p\mathbb{K}_s=h(s\mathbb{F}^e_p-\mathbb{K}_{sp})$} (17.5,8);
\draw [ fill={rgb,255:red,175; green,175; blue,175} ] (11.25,7.5) rectangle  node {\large $s\mathbb{F}^{**}_p\times p\mathbb{F}^{**}_s=h(\mathbb{D}_{sp})$} (17.5,5);

\node [font=\normalsize] at (3.5,12.5) {$h^{-1}(.,.)$};
\node [font=\normalsize] at (3.5,11.5) {$p\mathbb{F}_s$};
\node [font=\normalsize] at (4.5,12) {$s\mathbb{F}_p$};
\draw [ line width=0.2pt , dashed] (4.5,11.5) rectangle  (17.75,4.75);

\node [font=\normalsize] at (3.5,4.75) {$p\mathbb{F}_s$};
\node [font=\large] at (18.25,4.75) {$\mathbb{Z}_{sp}$};
\node [font=\normalsize] at (18,12.25) {$s\mathbb{F}_p$};
\end{tikzpicture}
}%
\caption{The three boxes at the top row represent the three subsets of $s\mathbb{F}_p$, while the three boxes at the leftmost column represent the three subsets of $p\mathbb{F}_s$. Each remaining box corresponds to a subset of $\mathbb{Z}_{sp}$ formed by combining the sets of its row and column.}
\label{fig:9_subsets}
\end{figure}

\medskip

In this paper, we proved an isomorphism between the finite fields $\mathbb{F}_s$, $p\mathbb{F}_s$ and $p\mathbb{F}^{+1}_s$ (and similarly $\mathbb{F}_p$, $s\mathbb{F}_p$, $s\mathbb{F}^{+1}_p$).  
We also showed that the Cartesian set $s\mathbb{F}_p \times p\mathbb{F}_s$ forms a ring with zero divisors.  

Next, we defined an isomorphism $h$ between $s\mathbb{F}_p \times p\mathbb{F}_s$ and the ring $\mathbb{Z}_{sp}$, which allows any element of $\mathbb{Z}_{sp}$ to be expressed uniquely as a linear combination of $s$ and $p$.  

Splitting $s\mathbb{F}_p$ into the sets $\{0\}$, $s\mathbb{K}_p$ and $s\mathbb{F}^{**}_p$ reveals nine subsets of $\mathbb{Z}_{sp}$, as depicted in Figure~\ref{fig:9_subsets}. These subsets were analyzed in Sections~\ref{sec:sec3} and~\ref{sec:sec4}.  

The sets $s\mathbb{K}_p$ and $p\mathbb{K}_s$ are multiplicative groups, producing the multiplicative group $\mathbb{K}_{sp}$, which is isomorphic to their Cartesian product and serves as the kernel of $\mathbb{Z}_{sp}$.  

The sets $s\mathbb{F}_p$ and $p\mathbb{F}_s$ contain multiples of $s$ and $p$, which are zero divisors. Since these elements do not have inverses in $\mathbb{Z}_{sp}$, we restrict the domain to allow for the computation of inverses.  

Furthermore, we defined the \emph{off-by-one} sets as the Cartesian products of $s\mathbb{F}^e_p$ and $p\mathbb{F}^e_s$. This is the domain of the Dickson polynomial, since every element must be invertible.  

Finally, we defined $\mathbb{D}_{sp}$, which contains elements of $\mathbb{Z}_{sp}$ that are neither multiples of $s$ or $p$ nor off-by-one. Elements in $\mathbb{D}_{sp}$, when repeatedly squared, always differ from multiples of $s$ or $p$ by at least 2. All these sets are illustrated in Figure~\ref{fig:9_subsets}.  

\medskip

The graph of $\mathbb{D}_{sp}$ is organized into cycles and trees, with cyclic elements forming the roots of their respective trees.  
The number of elements in each cycle of $\mathbb{D}_{sp}$ equals the least common multiple of the lengths of the cycles originating from $\mathbb{F}_s$ and $\mathbb{F}_p$.  

All elements that could be exploited in a cyclic attack~\cite{rivest1978remarks} are members of $\mathbb{D}_{sp}$. Thus, the security of RSA against cyclic attacks depends solely on the largest divisor of $(s-1)/2^k - 1$, where $s$ is the smaller of the two primes $s$ and $p$.  

Quadratic sieve factoring algorithms~\cite{pomerance1984quadratic} attempt to factor the product $sp$ by constructing numbers $x$ and $y$ such that $\gcd((y^2-x^2) \bmod sp, sp)$ yields either $s$ or $p$.  
The numbers $x$ and $y$ are square roots of a cyclic element $a$ and are located at the first level of the tree rooted at $a$.  

Therefore, $\mathbb{D}_{sp}$ is a promising set for further study, as analyzing its cycles and trees may offer insights for improving factorization algorithms.

\section{Conclusion}

In this work, we analyzed the structure of $\mathbb{Z}_{sp}$ by decomposing it into subsets related to $s\mathbb{F}_p$ and $p\mathbb{F}_s$, including their kernels and the off-by-one elements. It is well-known that the Cartesian product $s\mathbb{F}_p \times p\mathbb{F}_s$ forms a ring with zero divisors; we have included the explicit isomorphism $h$ between this ring and $\mathbb{Z}_{sp}$ to make the exposition self-contained.

Using this framework, we defined the set $\mathbb{D}_{sp}$, consisting of elements that are neither multiples of $s$ or $p$ nor off-by-one. We showed that $\mathbb{D}_{sp}$ naturally organizes into cycles and trees, with cyclic elements forming the roots, and we calculated the sizes of these structures explicitly.  

Finally, we highlighted the cryptographic relevance of $\mathbb{D}_{sp}$. Its structure captures elements that are resistant to cyclic attacks, and its trees and cycles could provide a framework for analyzing factorization techniques such as the quadratic sieve. This opens the possibility for further exploration of $\mathbb{D}_{sp}$ in the context of both algebraic number theory and cryptographic applications.

\nocite{*}
\bibliographystyle{fundam}
\bibliography{mk}

%%%%%%%%%%%%%%%%%%%%%%%%%%%%%%%%%%%%%%%%%%%%%%%%%%%%%%%%%%%%%%%%%%%%%%

\end{document}